\documentclass[sigconf,screen]{acmart}
\acmYear{2025}\copyrightyear{2025}
\acmConference[SPAA '25]{37th ACM Symposium on Parallelism in Algorithms and Architectures}{July 28--August 1, 2025}{Portland, OR, USA}
\acmBooktitle{37th ACM Symposium on Parallelism in Algorithms and Architectures (SPAA '25), July 28--August 1, 2025, Portland, OR, USA}
\acmDOI{10.1145/3694906.3743325}
\acmISBN{979-8-4007-1258-6/25/07}

\newif\iffull
\fulltrue 

\startPage{1}

\usepackage{math-cmds}
\usepackage[ruled,vlined,linesnumbered]{algorithm2e}
\usepackage[normalem]{ulem}
\usepackage{latexsym,amsthm,amsmath,amsfonts,stmaryrd}
\usepackage{listings}
\usepackage{caption}
\usepackage{xcolor}
\usepackage{colortbl}
\usepackage{enumitem}

\usepackage{centernot}
\usepackage{scalerel}

\usepackage[export]{adjustbox}
\usepackage{wrapfig}

\usepackage{multirow}

\usepackage{newtxmath}
\usepackage{subfig}
\usepackage[utf8]{inputenc}
\usepackage[T1]{fontenc}
\usepackage{mathpartir}
\usepackage{microtype}
\renewcommand{\infer}[3][]{
\ifthenelse{\equal{#1}{}}{
\inferrule{#2}{#3}
}
{
\inferrule*[right={\scriptsize \textbf{#1}}]
{#2}
{#3}
}
}

\usepackage{braket}
\usepackage{enumitem}
\setlist{itemsep=0.0ex,parsep=0pt}    

\newtheorem{theorem}{Theorem}

\newcommand{\ttt}[1]{\texttt{#1}}

\newcommand{\out}[1] {}

\newcounter{codeLineCntr}

\newcommand{\cfbox}[2]{%
    \colorlet{currentcolor}{.}%
    {\color{#1}%
    \fbox{\color{currentcolor}#2}}%
}

\reversemarginpar
\newif\ifnotes
\notestrue

\newcommand{\punt}[1]{}

\renewcommand{\eqref}[1]{Equation~(\ref{eq:#1})}

\newcommand{\abs}[1]{\left| #1\right|}

\newcommand{\ceil}[1]{\left\lceil #1 \right\rceil}

\newcounter{remark}[section]

\newcommand{\todo}[1]{}

\setcopyright{none}
\begin{document}

\title{POPQC: Parallel Optimization for Quantum Circuits}
\author{Pengyu Liu}
\email{pengyuliu@cmu.edu}
\affiliation{%
	\institution{Carnegie Mellon University}
	\city{Pittsburgh}
	\state{PA}
	\country{USA}
}

\author{Jatin Arora}
\email{jatina29@gmail.com}
\affiliation{%
	\institution{Carnegie Mellon University}
	\city{Pittsburgh}
	\state{PA}
	\country{USA}
}
\authornote{Currently at Amazon Web Services. This work was completed while at Carnegie Mellon University.}

\author{Mingkuan Xu}
\email{mingkuan@cmu.edu}
\affiliation{%
	\institution{Carnegie Mellon University}
	\city{Pittsburgh}
	\state{PA}
	\country{USA}
}

\author{Umut A. Acar}
\email{umut@cmu.edu}
\affiliation{%
	\institution{Carnegie Mellon University}
	\city{Pittsburgh}
	\state{PA}
	\country{USA}
}

\date{}

\thispagestyle{empty}
\newcommand{\project}{\textsf{PAQC}}
\newcommand{\PN}{\project}
\newcommand{\kwcost}[1]{\mathbf{cost}\left(  {#1} \right)}
\newcommand{\circuitcon}[2]{{#1} + {#2}}

\newcommand{\bigomega}{\mathbf{\Omega}}

\newcommand{\sfbox}[1]
{
\cfbox{blue}{#1}
}
\newcommand{\srule}{\vspace{2mm}\rule{\columnwidth}{1pt}\vspace{2mm}}

\newcommand{\lang}{\textsc{Laqe}}
\newcommand{\dom}[1]{\mathop{\text{dom}}(#1)}
\newcommand{\codom}[1]{\mathop{\text{cod}}(#1)}
\newcommand{\kw}[1]{\mbox{\ttt{#1}}}
\newcommand{\cdparens}[1]{({#1})}
\newcommand{\cd}[1]{{\lstinline!#1!}}
\newcommand{\hmm}{\textsf{HMM}}
\newcommand{\rulename}[1]{\textsc{#1}}
\newcommand{\ruleref}[1]{Rule~\rulename{#1}}

\newcommand{\true}{\ensuremath{\kw{true}}}
\newcommand{\false}{\ensuremath{\kw{false}}}
\newcommand{\ttrue}{\kw{t}}
\newcommand{\ffalse}{\kw{f}}
\newcommand{\prog}{\ensuremath{P}}
\newcommand{\pred}{\ensuremath{\mathcal{P}}}
\newcommand{\predf}[2]{\ensuremath{\pred(#1,#2)}}
\newcommand{\defeq}{\triangleq}

\newcommand{\type}[2]{\ensuremath{#1 : #2}}
\newcommand{\typed}[4]{\ensuremath{#1 \vdash_{#2} \type{#3}{#4}}}

\newcommand{\btype}{\beta}
\newcommand{\utype}{\theta}
\newcommand{\val}{v}
\newcommand{\uval}{u}
\newcommand{\name}{\textsf{POPQC}}

\newcommand{\oac}{\textsf{OAC}}
\newcommand{\algnameminus}{\textsf{OACMinus}}
\newcommand{\coam}{\textsf{OAC}}
\newcommand{\lopt}{\textsf{Lopt}}
\newcommand{\coamwith}[1]{\ensuremath{\mathsf{SOAM}[{#1}]}}
\newcommand{\queso}{{\textsf{Queso}}}
\newcommand{\voqc}{{\textsf{VOQC}}}
\newcommand{\pyzx}{{\textsf{PyZX}}}
\newcommand{\quartz}{{\textsf{Quartz}}}
\newcommand{\quartztool}{$\mathsf{Quartz}$}
\newcommand{\quesotool}{$\mathsf{Queso}$}
\newcommand{\feyntool}{\textsf{FeynOpt}}

\newcommand{\quartzt}[1]{\ensuremath{\mathsf{Quartz}_{\,#1}}}
\newcommand{\quesot}[1]{\ensuremath{\mathsf{Queso}_{\,#1}}}
\newcommand{\coamt}[1]{\ensuremath{\mathsf{SOAM}[#1]}}
\newcommand{\clifft}{Clifford+T}

\newcommand{\compat}[2]{\ensuremath{{#1} \mathbin{\scaleobj{1.2}{\diamond}} {#2}}}
\newcommand{\notcompat}[2]{\ensuremath{{#1} \mathbin{\scaleobj{1.2}{\centernot{\diamond}}} {#2}}}
\newcommand{\windowopt}[2]{{#2}~\textsf{\textbf{segment-optimal}}_{#1}}
\newcommand{\wopttext}{segment optimal}
\newcommand{\compressed}[1]{{#1}~\textsf{\textbf{compact}}}
\newcommand{\locallyopt}[2]{{#2}~\textsf{\textbf{locally-optimal}}_{#1}}
\newcommand{\qubits}[1]{\mathsf{qubits}({#1})}
\newcommand{\tmemprog}{memory-progress\xspace}
\newcommand{\tmempres}{memory-preservation\xspace}

\newcommand{\kwint}{\kw{int}}
\newcommand{\kwnat}{\kw{nat}}
\newcommand{\kwfut}{\kw{fut}}
\newcommand{\kwprod}[2]{\ensuremath{{#1} \times {#2}}}
\newcommand{\kwarr}[2]{\ensuremath{{#1} \ra {#2}}}
\newcommand{\kwloc}[1]{\ensuremath{{#1}~\kw{loc}}}
\newcommand{\kwref}[1]{\ensuremath{{#1}~\kw{ref}}}

\newcommand{\kwtid}{\kw{tid}}
\newcommand{\kwunit}{\kw{unit}}
\newcommand{\kwok}{\kw{ok}}

\newcommand{\heap}{H}
\newcommand{\spc}{H}
\newcommand{\empspc}{\emptyset}
\newcommand{\spaceext}[3]{{#1}[{#2} \mapsto {#3}]}

\newcommand{\catspace}{\uplus}
\newcommand{\catheap}{\uplus}

\newcommand{\heapun}[1]{\langle #1 \rangle}
\newcommand{\heapbi}[2]{\langle #1 ; #2 \rangle}
\newcommand{\heaptri}[3]{\langle #1 ; #2; #3 \rangle}
\newcommand{\heapquad}[4]{\langle #1 ; #2 ; #3 ; #4 \rangle}
\newcommand{\restctx}[2]{\ensuremath{#1 \upharpoonright_{#2}}}
\newcommand{\freeloc}[1]{\ensuremath{\mathsf{FL}(#1)}}
\newcommand{\locs}[1]{\ensuremath{\mathsf{Loc}(#1)}}
\newcommand{\diff}[1]{\ensuremath{\mathsf{Diff}(#1)}}

\newcommand{\rename}[3]{[#2 \mapsto #3](#1)}
\newcommand{\kwt}{\kw{t}}

\newcommand{\estore}{[~]}
\newcommand{\mkstore}[2]{\ensuremath{{#1}::{#2}}}

\newcommand{\kwn}{\kw{n}}
\newcommand{\kwlet}[3]{\kw{let}~{#1}={#2}~\kw{in}~{#3}~\kw{end}}
\newcommand{\kwfun}[3]{\ensuremath{\kw{fun}~{#1}~{#2}~\kw{is}~{#3}~\kw
{end}}}
\newcommand{\kwpair}[2]{\ensuremath{\langle{#1},{#2}}\rangle}
\newcommand{\kwapply}[2]{\ensuremath{{#1}~{#2}}}
\newcommand{\kwfst}[1]{\ensuremath{\kw{fst}\cdparens{#1}}}
\newcommand{\kwsnd}[1]{\ensuremath{\kw{snd}\cdparens{#1}}}
\newcommand{\gcing}[1]{\ensuremath{[#1]}}
\newcommand{\kwnew}[1]{\ensuremath{\kw{ref}(#1)}}
\newcommand{\kwderef}[1]{\ensuremath{\mathop{!}#1}}
\newcommand{\kwwrite}[2]{\ensuremath{#1 \mathop{:=} #2}}

\newcommand{\kwletrec}[2]{\ensuremath{#1 \mathop{\cdot} #2}}
\newcommand{\kwtask}[3]{\ensuremath{#1 \mathop{\cdot} #2 \mathop{\cdot} #3}}
\newcommand{\kwtaskalt}[2]{\ensuremath{#1 \mathop{\cdot} #2}}
\newcommand{\halt}{\bot}
\newcommand{\tree}{T}
\newcommand{\trace}{t}
\newcommand{\kwfork}[1]{{\ensuremath{\kw{fork}\cdparens{#1}}}}
\newcommand{\kwjoin}[1]{{\ensuremath{\kw{join}\cdparens{#1}}}}
\newcommand{\kwunitv}{\ensuremath{(\,)}}
\newcommand{\kwtidv}{\ensuremath{\kw{t}}}
\newcommand{\gheap}{G}
\newcommand{\lheap}{\heap}
\newcommand{\tolheap}[1]{\ensuremath{\Delta(#1)}}
\newcommand{\theap}[2]{\left(#1, #2\right)}

\newcommand{\cdpar}{\texttt{par}}
\newcommand{\kwpar}[2]
           {\ensuremath{\mathop{\vartriangleleft \hspace{-0.1em} #1, #2
               \hspace{-0.1em} \vartriangleright}}}
\newcommand{\kwpara}[2]
           {\ensuremath{\mathop{\blacktriangleleft \hspace{-0.1em} #1, #2
               \hspace{-0.1em} \blacktriangleright}}}
\newcommand{\kwparl}[2]{\ensuremath{#1 \overset{\leftarrow}{\|} #2}}
\newcommand{\kwparr}[2]{\ensuremath{#1 \overset{\rightarrow}{\|} #2}}
\newcommand{\task}{T}
\newcommand{\config}{\mathcal{C}}

\newcommand{\flate}[1]{\hat{#1}}
\newcommand{\flatten}[3]{\left\| #2 \right\|_{#1} \leadsto #3}
\newcommand{\fstep}{\step}

\renewcommand{\a}{\ensuremath{\alpha}}
\renewcommand{\b}{\ensuremath{\beta}}
\newcommand{\h}{\ensuremath{\eta}}
\renewcommand{\r}{\ensuremath{\rho}}
\newcommand{\p}{\ensuremath{P}}
\newcommand{\s}{\p}
\newcommand{\om}{\ensuremath{\Omega}}
\renewcommand{\l}{\ensuremath{l}}
\newcommand{\sig}{\ensuremath{\Sigma}}
\newcommand{\empctx}{\ensuremath{\cdot}}

\newcommand{\la}{\leftarrow}
\newcommand{\ra}{\rightarrow}
\newcommand{\pstep}{\Rightarrow}
\newcommand{\tstep}{\Rightarrow}
\newcommand{\optstep}{\longmapsto}
\newcommand{\compstep}{\longmapsto_{\delta}}
\newcommand{\localstep}[3]{{#2} \overset{#1}{\optstep} {#3}}
\newcommand{\globstep}[2]{{#1} \compstep {#2}}
\newcommand{\step}{\ra}
\newcommand{\stepgc}[1]{\xra[{\mbox{\tiny GC}}]{#1}}
\newcommand{\gcstep}{\ra_{\mbox{\tiny GC}}}
\newcommand{\pgcstep}{\pstep_{\mbox{\tiny GC}}}
\newcommand{\cgcstep}{\rightarrow_{\mbox{\tiny CGC}}}


\newcommand{\fresh}{\ensuremath{\; \mathsf{fresh}}}
\newcommand{\starrow}[1]{\xrightarrow{#1}}
\newcommand{\alloc}[4]{#1; #2 \starrow{alloc} #3; #4}
\newcommand{\update}[4]{#1; #2; #3 \starrow{update} #4}
\newcommand{\lookup}[3]{#1; #2 \starrow{lookup} #3}
\newcommand{\newtask}[2]{#1 \starrow{new} #2}
\newcommand{\isdone}[3]{#1 \starrow{done} #2; #3}
\newcommand{\diffs}[1]{\mathit{diff}(#1)}
\newcommand{\initial}{\ensuremath{\;\mathsf{initial}}}
\newcommand{\htyped}[3]{\vdash_{#3} #1 : #2}

\newcommand{\heaptype}[3]{\left(#1, #2\right) : #3}
\newcommand{\allocg}[4]{\mathit{allocg}\left(#1, #2\right) = \left(#3, #4\right)}
\newcommand{\allocl}[4]{\mathit{allocl}\left(#1, #2\right) = \left(#3, #4\right)}
\newcommand{\promote}[6]{#1; #2; #3 \starrow{promote} #4; #5; #6}
\newcommand{\promotebrl}[3]{#1; #2; #3}
\newcommand{\promotebra}{\starrow{promote}}
\newcommand{\promotebrr}[3]{#1; #2; #3}
\newcommand{\pmap}{M}
\newcommand{\greachable}[1]{\mathsf{greachable}\left(#1\right)}

\newtheorem{thm}[theorem]{Theorem}
\newtheorem{lem}[theorem]{Lemma}

\newenvironment{rulearray}
{
\newcommand{\newcol}{\qquad}
\newcommand{\newcolhalf}{\quad}
\newcommand{\newrow}{\\[4ex]}
\newcommand{\newrowhalf}{\\[2ex]}
\[
\begin{array}{c}
}
{
\end{array}
\]
\let\newcol\undefined
\let\newrow\undefined
}

\newcommand{\ramtodo}[2][]
{\todo[color=magenta,author=Ram,size=\small,#1]{#2}}

\newcommand{\defn}[1]{\emph{\textbf{#1}}}
\newcommand{\mpl}{\textsf{MPL}}

\newcommand{\rulereftwo}[2]{rules~\rulename{#1} and \rulename{#2}}
\newcommand{\with}{\ensuremath{\mathbin;}}

\newcommand{\highlight}[1]{\colorbox{gray!20}{\ensuremath{#1}}}
\newcommand{\hred}[1]{\colorbox{red!10}{\ensuremath{#1}}}
\newcommand{\hblue}[1]{\colorbox{blue!10}{\ensuremath{#1}}}
\newcommand{\hgreen}[1]{\colorbox{green!10}{\ensuremath{#1}}}
\newcommand{\sizeof}[1]{\ensuremath{\lvert #1 \rvert}}
\newcommand{\costof}[1]{\ensuremath{\mathbf{cost} ({#1})}}
\newcommand{\cost}{\ensuremath{\mathbf{cost}}}
\newcommand{\oracle}{\ensuremath{\mathbf{oracle}}}

\definecolor{darkblue}{HTML}{0007C9}
\newcommand{\mh}[1]{{\ensuremath{\mbox{\ensuremath{#1}}}}}

\newcommand{\ctxemp}{\ensuremath{\cdot}}
\newcommand{\ctxext}[3]{\ensuremath{#1,#2\!:\!#3}} 
\newcommand{\etyped}[4]{\ensuremath{{#1} \vdash_{#2} {#3} : {#4}}}
\newcommand{\memtyped}[3]{\ensuremath{{#1} \vdash {#2} : {#3}}}
\newcommand{\gtyped}[3]{\ensuremath{{#1} \vdash {#2} : {#3}}}
\newcommand{\httyped}[6]{\ensuremath{{#1} \with {#2} \with {#3} \vdash {#4}\!\cdot\!{#5} : {#6}}}
\newcommand{\ttyped}[5]{\ensuremath{{#1} \with {#2} \with {#3} \vdash {#4} : {#5}}}

\newcommand{\sttyped}[6]{\ensuremath{{\vdash_{#1} {#2} \with {#3} \with {#4} \with {#5} : {#6}}}}
\newcommand{\getyped}[6]{\ensuremath{{#1} \vdash_{#2, #3} {#4} \with {#5} : {#6}}}

\newcommand{\typnat}{\kw{nat}}
\newcommand{\typint}{\kw{int}}
\newcommand{\typbool}{\kw{bool}}
\newcommand{\typchar}{\kw{char}}
\newcommand{\typfloat}{\kw{float}}
\newcommand{\typprod}[2]{\ensuremath{{#1} \times {#2}}}
\newcommand{\typfun}[2]{\ensuremath{{#1}\!\rightarrow\!{#2}}}
\newcommand{\typref}[1]{\ensuremath{{#1}~\kw{ref}}}
\newcommand{\typfut}[1]{\ensuremath{{#1}~\kw{fut}}}
\newcommand{\futs}[1]{\mathsf{Fut}(#1)}
\newcommand{\futsmem}[2]{\mathsf{Fut}(#1, #2)}

\newcommand{\enat}[1]{\ensuremath{#1}}
\newcommand{\efun}[3]{\ensuremath{\kw{fun}~{#1}~{#2}~\kw{is}~{#3}}}
\newcommand{\epair}[2]{\ensuremath{\langle {#1}, {#2} \rangle}}
\newcommand{\eapp}[2]{\ensuremath{{#1}~{#2}}}
\newcommand{\efst}[1]{\ensuremath{\kw{fst}~{#1}}}
\newcommand{\esnd}[1]{\ensuremath{\kw{snd}~{#1}}}
\newcommand{\eref}[1]{\ensuremath{\kw{ref}~{#1}}}
\newcommand{\ebang}[1]{\ensuremath{\mathop{!}#1}}
\newcommand{\eupd}[2]{\ensuremath{#1 \mathop{:=} #2}}
\newcommand{\elet}[3]{\kw{let}~{#1}={#2}~\kw{in}~{#3}}
\newcommand{\epar}[2]{\ensuremath{\langle {#1}\mathbin\|{#2} \rangle}}

\newcommand{\purelang}{{\sc $\lambda^{P}$}}
\newcommand{\reflang}{{\sc $\lambda^{U}$}}


\newcommand{\mememp}{\emptyset}
\newcommand{\memext}[3]{\ensuremath{#1}[{#2} \!\hookrightarrow\! {#3}]}

\newcommand{\actarrow}{\blacktriangleright}
\newcommand{\pasarrow}{\vartriangleright}
\newcommand{\fmap}{\Delta}
\newcommand{\femp}{\emptyset}
\newcommand{\fmapactive}[3]{\ensuremath{#1} [{#2} \!\actarrow\! {#3}]}
\newcommand{\fmapjoined}[3]{\ensuremath{#1} [{#2} \!\pasarrow\! {#3}]}

\newcommand{\futctxt}{\Knownctxt}
\newcommand{\Futctxt}{\Knownctxt}
\newcommand{\ReadLocs}{\mathsf{R}}
\newcommand{\Knownctxt}{{K}}
\newcommand{\fut}[2]{\kw{fut}(#1; #2)}
\newcommand{\harpfut}[1]{\kw{fut}(#1)}
\newcommand{\futctxtemp}{\emptyset}
\newcommand{\te}[1]{\{#1\}}
\newcommand{\hemp}{\emptyset}
\newcommand{\hcat}{\cup}
\newcommand{\hext}[2]{{#1},{#2}}

\newcommand{\tack}{\oplus}
\newcommand{\plug}{\bowtie}

\newcommand{\omparam}{step length}
\newcommand{\actwrite}[2]{\textbf{U}{#1}\!\Leftarrow\!{#2}}
\newcommand{\actalloc}[2]{\textbf{A}{#1}\!\Leftarrow\!{#2}}
\newcommand{\actread}[2]{\textbf{R}{#1}\!\Rightarrow\!{#2}}
\newcommand{\actsync}[2]{\textbf{F}{#1}\!\Rightarrow\!{#2}}
\newcommand{\actnone}{\textbf{N}}

\newcommand{\stepstar}{\longmapsto^*}
\newcommand{\tstepstar}{\tstep^*}
\newcommand{\drfstep}[2]{\xmapsto[{#2}]{{\,#1\,}}}
\newcommand{\drfstepstar}[1]{\xmapsto{{\,#1\,}}\joinrel\mathrel{^*}}

\newcommand{\gt}[2]{\ensuremath{\mathsf{GT}({#1},{#2})}}
\newcommand{\gemp}{\bullet}
\newcommand{\gseq}[2]{{#1}\oplus{#2}}
\newcommand{\gseqnamed}[3]{{#1}\oplus_{#2}{#3}}
\newcommand{\gseqa}[2]{\gseqnamed{#1}{a}{#2}}
\newcommand{\gseqb}[2]{\gseqnamed{#1}{b}{#2}}
\newcommand{\gspawn}[1]{\mathsf{spawn}\ {#1}}
\newcommand{\gsync}[1]{\mathsf{sync}\ {#1}}
\newcommand{\ghead}[1]{\mathsf{hd}(#1)}
\newcommand{\gtail}[1]{\mathsf{tl}(#1)}

\newcommand{\fcpar}[3]{\ensuremath{\gseq{#1}{(\gpar{#2}{#3})}}}

\newcommand{\gmerge}[2]{\bowtie_F ({#1}, {#2})}
\newcommand{\gmergerel}[3]{\bowtie_R ({#1}, {#2}) \downarrow {#3}}
\newcommand{\gcseq}[1]{\ensuremath{[#1]}}
\newcommand{\gcpar}[3]{\ensuremath{\gseq{#1}{(\gpar{#2}{#3})}}}
\newcommand{\gcparnamed}[4]{\ensuremath{\gseq{#1}{({#2}\otimes_{#3}{#4})}}}
\newcommand{\gcspawn}[4]{\ensuremath{\gseq{#1}{\gseq{#2}{(\gpar{#3}{#4})}}}}

\newcommand{\gpar}[2]{{#1}\otimes_{a}{#2}}
\newcommand{\gw}[1]{\ensuremath{\mathsf{W}({#1})}}
\newcommand{\ga}[1]{\ensuremath{\mathsf{A}({#1})}}
\newcommand{\greads}[1]{\ensuremath{\ReadLocs({#1})}}
\newcommand{\gaw}[1]{\ensuremath{\mathsf{AW}({#1})}}
\newcommand{\lw}[1]{\ensuremath{\mathsf{LW}({#1})}}
\newcommand{\alw}[1]{\ensuremath{\mathsf{A}({#1}) \cup \mathsf{LW}({#1})}}
\newcommand{\gabw}[1]{\ensuremath{\gaw{#1}}}

\newcommand{\saw}[1]{\ensuremath{\mathsf{SP}({#1})}}
\newcommand{\gf}[1]{\ensuremath{\overline{#1}}}

\newcommand{\extendsfj}[2]{\ensuremath{{#1}~\textsf{extends}~{#2}~\textsf{with f/j}}}
\newcommand{\extendswith}[3]{\ensuremath{{#1}~\textsf{extends}~{#2}~\textsf{with}~{#3}}}

\newcommand{\geok}[2]{{#1} \with {#2}~\textit{ok}}
\newcommand{\loc}[1]{{#1}~\textit{loc}}
\newcommand{\gleaf}[1]{{#1}~\textit{leaf}}
\newcommand{\gnode}[1]{{#1}~\textit{node}}

\newcommand{\drf}[2]{{#1} \vdash {#2}~{\textit{drf}}}
\newcommand{\drfb}[2]{{#1} \vdash {#2}~{\textit{wrf}}}
\newcommand{\drft}{\textit{drf}}

\theoremstyle{plain}
\newtheorem{property}{Property}


\newcommand{\flushLR}[3]{\hspace*{#3}\makebox[0em][l]{#1}\hspace*{\fill}\makebox[0em][r]{#2}\hspace*{#3}}
\newcommand{\rulesdesc}[2]{\textbf{#1}\hspace*{1em}{\fbox{#2}}}
\newcommand{\desc}[1]{\textbf{#1}}

\newdimen\zzlistingsize
\newdimen\zzlistingsizedefault
\zzlistingsizedefault=9pt
\newdimen\kwlistingsize
\kwlistingsize=9pt
\zzlistingsize=\zzlistingsizedefault
\gdef\lco{black}

\begin{abstract}

	Optimization of quantum programs or circuits is a fundamental
  problem in quantum computing and remains a major challenge.
%
  State-of-the-art quantum circuit optimizers rely on heuristics and
  typically require superlinear, and even exponential, time.
	Recent work~\cite{oac2025} proposed a new approach that pursues a weaker form
	of optimality called local optimality.
  Parameterized by a natural number $\Omega$, local optimality insists
  that each and every $\Omega$-segment of the circuit is optimal with
  respect to an external optimizer, called the \emph{oracle}.
	%
	%
	Local optimization can be performed using only a linear number of
  calls to the oracle but still incurs quadratic computational
  overheads in addition to oracle calls.
	Perhaps most importantly, the algorithm is sequential.

	In this paper, we present a parallel algorithm for local optimization of
	quantum circuits.
	To ensure efficiency, the algorithm operates by keeping a set of
  \emph{fingers} into the circuit and maintains the invariant that a
  $\Omega$-deep circuit needs to be optimized only if it contains a
  finger.
	Operating in rounds, the algorithm selects a set of fingers,
  optimizes in parallel the segments containing the fingers, and
  updates the finger set to ensure the invariant.
	For constant $\Omega$, we prove that the algorithm requires $O(n\lg{n})$ work
	and $O(r\lg{n})$ span, where $n$ is the circuit size and $r$ is the number of
	rounds.
	We prove that the optimized circuit returned by the algorithm is locally
	optimal in the sense that any $\Omega$-segment of the circuit is optimal
	with respect to the oracle.

	To assess the algorithm's effectiveness in practice, we implement it in the
	Rust programming language and evaluate it by considering a range of quantum benchmarks and
	several state-of-the-art optimizers.
	The evaluation shows that the algorithm is work efficient and scales
  well as the number of processors (cores) increases.
	On our benchmarks, the algorithm outperforms existing optimizers, which are all
	sequential, by as much as several orders of magnitude, without degrading the
	quality.
	Our code is available on GitHub at \url{https://github.com/UmutAcarLab/popqc}.

\end{abstract}



\maketitle
\section{Introduction}

Quantum computing is an emerging field at the intersection of computer science,
physics, and mathematics.
By using quantum bits or qubits that offer exponential advantage by operating
in a superposition of states, quantum computers hold the promise of major
breakthroughs
in numerous fields including
quantum simulation\cite{Feynman82,Benioff80},
optimization\cite{childs2017quantum, peruzzo2014variational},
cryptography\cite{shor1994algorithms}, and
machine learning\cite{biamonte2017quantum, schuld2015introduction}.
Although quantum computer hardware has made significant advances in recent
years, including superconducting\cite{kjaergaard2020superconducting}, trapped
ions\cite{monroe2021programmable, moses2023race}, and Rydberg atom
arrays\cite{ebadi2021quantum, scholl2021quantum}, they still suffer from
numerous limitations, including a natural tendency to decohere, or
imperfections in the quantum gates, which can lead to accumulating errors or
``noise''.
Due to the modest scale and noisy nature of quantum computers, the modern and
near-term quantum era is sometimes referred to as the NISQ (Noisy
Intermediate-Scale Quantum) era\cite{preskill2018quantum}.
%


In this context, optimization of quantum programs or circuits has emerged as an
important area of research.
By optimizing the quantum circuit, circuit optimizers can reduce the number of
gates and thus the number of operations in the circuit.
In addition to increasing efficiency, quantum optimizers can improve circuit
fidelity by reducing errors due to decoherence, and even make computations
possible that are otherwise impossible (by bringing the computation within the
decoherence envelope of the quantum computer).

Optimization of quantum circuits, however, is a challenging problem:
global optimization is QMA-hard\cite{Nam_2018}, which is believed to be beyond
the reach of even quantum computers.
Practical optimizers therefore typically rely on heuristics.
For example, Nam et al.'s rule-based optimization\cite{Nam_2018} and Hietala's
verified implementation\cite{hietala2021verified} lack quality guarantees and
require quadratic time in the size of the circuit.
Xu et al.'s automated search-based optimizations\cite{quartz-2022,queso-2023}
face severe scalability issues because they rely on an exponential-time search
algorithm to apply optimization rules.
For these reasons, state-of-the-art optimizers can struggle to optimize
larger quantum circuits within a reasonable amount of time (e.g., within
hours).
As quantum computing continues to move from the NISQ era to the FASQ (Fault-Tolerant Application-Scale Quantum) era\cite{preskill2024beyond},
the efficiency and quality limitations of quantum optimizers become even more
important, because quantum circuits are expected to contain millions of gates.
Even as the scale of the circuit optimization challenge increases, quantum
circuit optimization approaches remain sequential.
We don't know of any provably and practically efficient parallel algorithms for
this task.

Recently, Arora et al. proposed an approach to quantum circuit optimization
that is able to offer both quality guarantees and efficiency\cite{oac2025}
with respect to a given \emph{oracle} optimizer.
Their algorithm assumes that the oracle optimizer works well for small to
moderate circuits of up to a few thousand gates and is parameterized by a
natural number $\Omega$.
The algorithm cuts the circuit into $\Omega$-segments, a block containing $\Omega$ gates, optimizes each
segment by using the oracle, and melds the optimized segments by optimizing
along the seams of the cuts.
The algorithm then compresses the circuit by moving all gates to the beginning
of the circuit as much as possible, thus minimizing the gaps in the circuit,
which can reduce the effectiveness of the optimizer.
By repeating this cut-optimize-meld-compress process until convergence, Arora
et al.\cite{oac2025} prove that their algorithm can guarantee that the output
circuit has the local optimality property, which ensures that any $\Omega$-wide
segment in the circuit is optimal with respect to the oracle.
The key to this property is the meld algorithm that ``propagates''
optimizations in one segment to adjacent segments.
They also prove, under some reasonable assumptions, that the algorithm makes a
linear number of calls to the oracle and show that the algorithm can improve
performance significantly, leading to about an order of magnitude improvement
without degrading quality.
Although Arora et al.'s work has made significant progress by showing a path to
improving the efficiency of quantum circuit optimization, their algorithm is
sequential.
Notably, the meld operations are inherently sequential, as they propagate
optimizations from one segment to the other by sequentially optimizing segments
in a sliding-window style.
Furthermore, their algorithm incurs quadratic overheads to implement the cut,
meld, and compress operations on circuits.

In this paper, we present an efficient parallel algorithm for local
optimization.
As with Arora et al.'s algorithm, our algorithm relies on an external oracle to
optimize circuit segments and takes a parameter $\Omega$.
Unlike their algorithm, our algorithm avoids cut and meld operations for
reasons of efficiency and parallelism.
The algorithm instead keeps a set of \emph{fingers} into the circuit and
maintains the invariant that all unoptimized $\Omega$-segments of the
circuit contain a finger.
Operating in rounds, the algorithm selects a set of non-interfering fingers,
optimizes in parallel the segments containing the fingers, and updates the
fingers to ensure the invariant.
To expose parallelism, the algorithm only uses a parallel-map construct, which
may be implemented in many ways, e.g., by using fork-join primitives.

To support efficient parallel access and updates to the quantum circuit, we
represent the circuit with a sparse array of gates that allows gates to be
removed, and pair it with an \emph{index tree} that allows finding all
non-deleted gates.
Using this representation, we prove that the total work (uniprocessor time) of
the algorithm is $O(n(\Omega\lg{n}+W))$, where $n$ is the number of gates in the
input circuit and $W$ is the work of the oracle on $2\Omega$-segments.
In practice, $\Omega$ is a moderate constant (in hundreds), leading to constant
work $W$ for the oracle, and $O(n\lg{n})$ work for the optimizer.
For span (parallel time), we show a bound $O(r(\lg{n}+S))$, where $r$ is the
number of rounds required by the algorithm, $n$ is the number of gates in the
input circuit, and $S$ is the span of the oracle optimizer (equals to $W$ for sequential
oracles).
This bound shows that the algorithm delivers significant parallelism in each
round, with the caveat that total parallelism may be limited if the number of
rounds is large.
Fortunately, in practice, we observe that the number of rounds is a modest
number (fewer than $100$ in most of our experiments) and the algorithm exhibits
significant parallelism.

To establish a quality guarantee on the optimized circuits, we prove that the
optimized circuit returned by the algorithm is locally optimal, meaning any
$\Omega$-segment of the circuit is optimal with respect to the oracle.
This shows that our approach can guarantee some degree of quality of the
output while also ensuring efficiency and parallelism.

To assess the practicality of the algorithm, we present an implementation in
the Rust language and evaluate the algorithm by using a number of quantum
benchmarks.
The evaluation shows that the number of rounds required by the algorithm is
small relative to the circuit size, and therefore the algorithm scales well to
multiple cores in practice, especially as the input circuits grow larger.
The evaluation also shows that the constant factors hidden in our asymptotic
analysis are small and that the algorithm is fast in practice: single-processor
runs of our algorithm outperform state-of-the-art sequential optimizers.
As a combined effect of the small constant factors and parallelism, our
optimizer delivers orders of magnitude speedups over existing optimizers,
especially as circuit sizes grow, while incurring no noticeable degradation in
the quality of optimization.
Notably, on a 64-core computer, our parallel optimizer can optimize circuits in
seconds that existing optimizers are unable to optimize within 24 hours of
compute time.
This result shows that parallel optimization of quantum circuits can be
effective, both in terms of performance and quality.

The specific contributions of the paper include the following:
\begin{itemize}
	\item A parallel algorithm for optimizing quantum circuits.
	\item Bounds on the work and span of the algorithm.
	\item Proof that the algorithm returns a locally optimal circuit.
	\item An implementation of the algorithm in the Rust language.
	\item An evaluation of the algorithm by considering multiple oracle optimizers and
	      challenging quantum benchmarks.
\end{itemize}
\section{Background}
\label{sec:background}
In this section, we introduce basic concepts of quantum computing and establish the notation used throughout this paper.
\subsection{Quantum States}
A quantum bit or a \defn{qubit} is the basic unit of quantum computation. The
state of a qubit $\ket{\psi}$ can be represented as a linear superposition of
basis states: $\ket{\psi} = \alpha\ket{0} + \beta\ket{1}$, where $\alpha,\beta
	\in \mathbb{C}$ satisfy the normalization condition $|\alpha|^2 + |\beta|^2 =
	1$. We can also write it as a vector for mathematical manipulation: $\begin{bmatrix} \alpha & \beta \end{bmatrix}$. For an $n$-qubit system, the quantum state becomes a superposition of $2^n$ basis states:
$\ket{\psi} = \sum_{i\in\{0,1\}^n} \alpha_i \ket{i}$, and can be written as a $2^n$-dimensional vector.

A quantum algorithm realizes a unitary transformation $U$, a $2^n\times 2^n$
matrix with $UU^\dagger = I$. When applied to the initial state $\ket{0}$, the
state becomes $U\ket{0} = \ket{\psi}$, which contains the result of the
algorithm.

\subsection{Gates and Circuits}
\label{sec:bg:circuits}
A unitary is neither implementable on a quantum computer nor representable due
to its exponential size. So, we use a quantum circuit $\mathcal{C}$ to describe
a quantum algorithm, which can be represented as a sequence of gates. A gate
only acts on a small number of qubits (usually one or two), and the other
qubits are left unchanged. We use $[g]$ and $[\mathcal{C}]$ to denote the
matrix representation of a gate $g$ and a circuit $\mathcal{C}$, respectively.
For a single-qubit gate $g$ that acts on qubit $i$, the matrix representation
is $[g]=I^{\otimes i} \otimes U \otimes I^{\otimes n-i-1}$, where $n$ is the
total number of qubits and $U$ is a $2\times 2$ unitary matrix. The matrix
representation of the whole circuit $\mathcal{C}:g_1,g_2,\ldots,g_k$ is the
product of the matrix representations of all the gates: $[\mathcal{C}] = [g_k][g_{k-1}] \ldots [g_1]$. From the matrix representation, we can see that a
fundamental property of quantum circuits is that any subcircuit is
interchangeable with any equivalent subcircuit (those implementing the same
unitary transformation) because matrix multiplication is associative.

The most straightforward representation of a quantum circuit is the gate
sequence representation, where a circuit is represented as a sequence of gates.
Other representations also exist: for example, the layered representation,
where circuits are organized into layers of independent gates. Two gates are
defined as \defn{independent} if they act on disjoint sets of qubits. This
layered representation is valuable because it directly maps to parallel
execution schedules on quantum hardware, with the number of layers or
\defn{depth} serving as a natural indicator of the running time of the circuit,
which is a quantum analogue of span in classical computing.

\subsection{Quantum Circuit Optimization}
Circuit optimization transforms an input quantum circuit into an equivalent
circuit (same unitary) that minimizes some cost function, with some examples
being the number of gates, the circuit depth, the number of non-Clifford gates,
the number of two-qubit gates, etc.

The optimization of large quantum circuits presents significant challenges due
to the exponential growth in the dimensionality of the underlying unitary
transformations with increasing qubit count. Indeed, global optimization of
quantum circuits has been proven to be QMA-hard\cite{janzing2003identity}. As
a result, most existing quantum circuit optimization methods are local, meaning
that they only optimize a small part of the circuit at a time.

\subsection{Parallelism Model}

We specify the parallel algorithms using traditional algorithmic style, written
in pseudocode, and use traditional work and span (depth) analysis
(e.g.,\cite{ab-book-algorithms}).
We use a parallel-map (\textbf{parmap}) primitive as the only means of exposing
parallelism.
This primitive maps over the elements of a collection in parallel and
computes some value for each element, returning the collection of
result values.
When analyzing work and span, we assume that this primitive adds a
logarithmic (in the number of iterations) cost to the span of each
iteration.
This is a conservative assumption and it is realistic for the multicore
architecture, where parallel maps can be implemented with a
logarithmic-depth fork-join tree.
We note that for our specific algorithm and analysis, a stronger assumption
such as constant-span parallel-map would not improve our bounds, because the
iterations all have at least logarithmic span.

\section{A Parallel Data Structure for Quantum Circuits}

To optimize a quantum circuit in parallel, we propose a specialized data
structure that enables efficient parallel access and manipulation of gates. Our
data structure addresses a key challenge in quantum circuit optimization: as
the optimization process progresses, the circuit becomes increasingly sparse
due to gate removals. This sparsity requires efficient mechanisms for locating
neighboring gates without scanning the entire circuit.
%
%
%
%
\begin{algorithm}[ht]
	\caption{The interface for the quantum circuit data structure. For cost bounds, $n$ is the number of gates in the circuit.}
	\label{alg:circuit-ds}
	\SetKwBlock{Interface}{Interface Circuit}{end}
	\Interface{
		\kw{type circuit}
		\BlankLine
		\tcp{Create a circuit from a gate array.}
		\tcp{Cost: $O(n)$ work and $O(\lg{n})$ span}
		\textbf{def} create(g: gate array): circuit\;
		\BlankLine
		\tcp{Return the number of gates, excluding tombstones, before an index $i$.}
		\tcp{Cost: $O(\lg{n})$ work and span}
		\textbf{def} before(c: circuit, i: int): int\;
		\BlankLine
		\tcp{Return the $i^{th}$ gate, excluding all tombstones.}
		\tcp{Cost: $O(\lg{n})$ work and span}
		\textbf{def} get(c: circuit, i: int): gate\;
		\BlankLine
		\tcp{Replace each gate at the specified index with the specified gate.}
		\tcp{Use a tombstone to indicate a removed gate.}
		\tcp{Cost: For $l$ gates, $O(l \lg{n})$ work and $O(\lg{n})$ span}
		\textbf{def} substitute(c: circuit, g: (int$\times$gate) array): void\;
		\BlankLine
		\tcp{Return the gate array of the circuit, excluding all tombstones.}
		\tcp{Cost: $O(n)$ work and $O(\lg{n})$ span}
		\textbf{def} gates(c: circuit): gate array\;
	}
\end{algorithm}
Algorithm~\ref{alg:circuit-ds} shows the interface for our circuit data
structure along with the cost bounds for each operation.
To support the operations with the given bounds, we use an array to store the
gates.
This enables constant-time access to each gate given its index in the array.
To remove a gate (e.g., during optimization), we replace it with a
``tombstone'', which indicates an absent gate.

Our algorithm operates by optimizing circuit \defn{segments} using the oracle
optimizer. A segment is defined as a contiguous sequence of non-tombstone gates
between indices $i$ and $j$ in the circuit's gate array. We call a segment an
\defn{$\Omega$-segment} if it contains $\Omega$ gates. These segments form the
basic units of optimization in our approach.
For this to work efficiently, the algorithm must disregard the tombstones when
partitioning the circuit into segments, which are then optimized independently.
This can be challenging because as the optimization proceeds, the number of
tombstones increases.
We therefore augment the array with a binary tree data structure, which we call
the \defn{index tree}, that helps locate the gates efficiently.

The index tree takes the form of a complete binary tree, where each leaf of the
tree corresponds to a gate in the circuit and is labeled with a weight of $1$
if there is a gate at that position, and $0$ if there is a tombstone.
We tag each internal node with a weight, which is equal to the sum of the
weights of its children, indicating the number of gates in the subtree
rooted at that node.
Figure~\ref{fig:ftree} shows an example index tree
for a circuit with 5 gates.
Initially, as shown in Figure~\ref{fig:ftree}a, the index tree has 5 leaves,
all with a weight of $1$, representing the state before any optimization.
We realize that we can optimize the circuit by removing the two \texttt{X}
gates separated by a \texttt{CNOT} gate and replacing them with tombstones.
To update the index tree, we change the weights of the leaves corresponding to
the removed \texttt{X} gates to $0$ and update all the weights to reflect this
change.
Figure~\ref{fig:ftree}b shows the index tree after optimization.

\begin{figure}[t]
	\centering
	\includegraphics[width=\columnwidth]{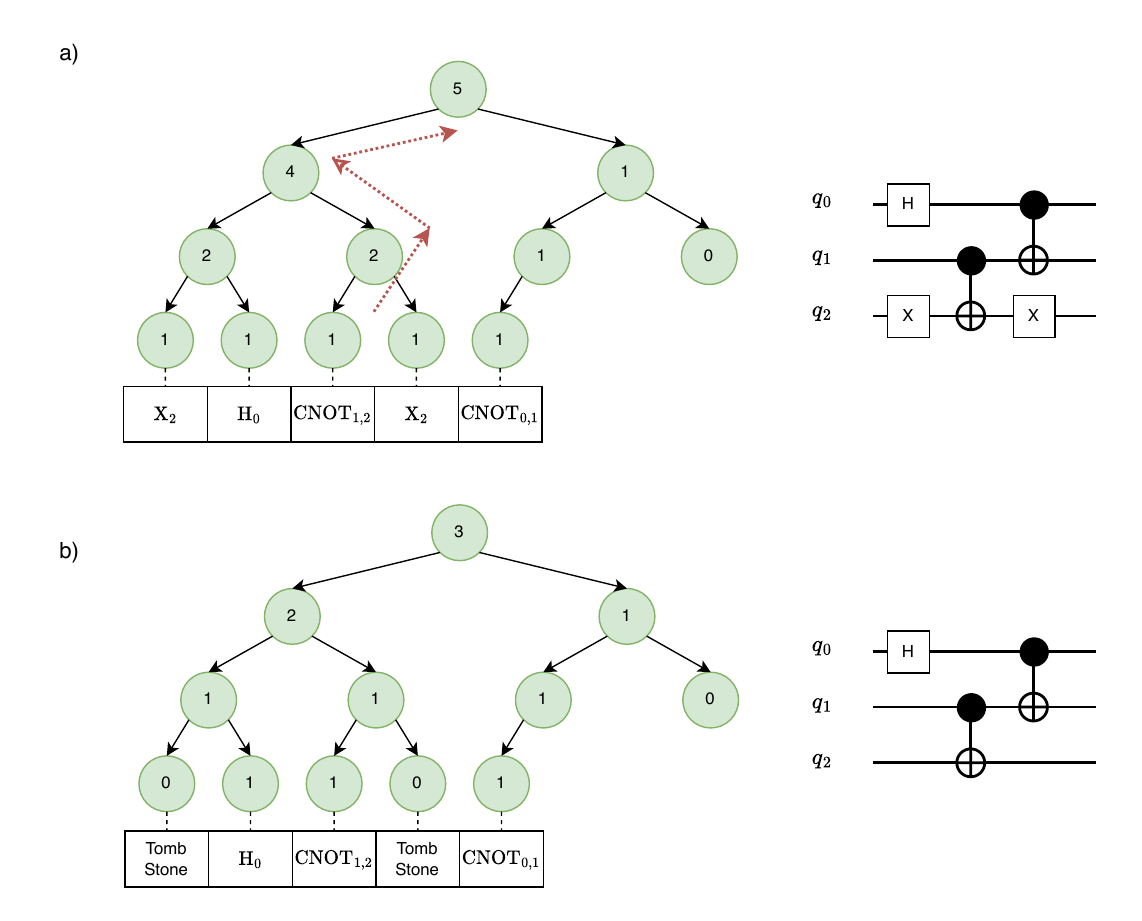}
	\caption{The index tree data structure. a) A circuit with 5 gates and a corresponding index tree. b) A circuit after optimization with 3 gates remaining and the updated index tree.}
	\label{fig:ftree}
\end{figure}

To \textbf{create} a circuit from a given array of gates, we construct the
index tree and tag each internal node layer by layer; this requires linear work
and logarithmic span.

For the \textbf{before} operation, we start at the leaf corresponding to the
specified index and walk up the index tree to the root, summing the weights
of left siblings along the path. Figure~\ref{fig:ftree}a illustrates this process with an example that finds the number of non-tombstone gates before $\texttt{CNOT}_{1,2}$. We trace the unique path from the leaf corresponding to $\texttt{CNOT}_{1,2}$ to the root (shown as the red path). Along this path, the first node has weight $1$ with no left sibling, while the second node has weight $2$ with a left sibling of weight $2$. Since only the second node contributes a left sibling weight, we conclude that there are two non-tombstone gates before $\texttt{CNOT}_{1,2}$.

Because the index tree is balanced, this requires $O(\lg{n})$ work and span.
The \textbf{get} operation takes an integer $i$ and fetches the $i^{th}$ gate,
ignoring all tombstones.
This can be implemented in logarithmic work and span by starting at the root of
the index tree and tracing a path down to a leaf.
%
The \textbf{substitute} operation takes an array of index-gate pairs,
substitutes the gates at the specified indices with the specified gates, and
updates the index tree.
The \textbf{gates} operation returns the gate array of the circuit, excluding
all tombstones.

The index tree data structure naturally generalizes to the layered
representation of circuits (discussed in Section~\ref{sec:bg:circuits}): we
think of each layer as a ``big'' gate and perform all operations at the
granularity of layers accordingly.
For our results, we primarily use the gate sequence representation, but we also
use the layered representation for an additional experiment on optimizing depth
(Section~\ref{sec:exp::quartz}).

\section{Algorithm}
\label{sec:algorithm}

We present the \name{} (Parallel Optimizer for Quantum Circuits) algorithm in
this section.
At a high level, the algorithm works by tracking a set of \defn{fingers} to be
optimized and optimizing in parallel the segments around the fingers until no
further optimization is possible. The fingers are a set of indices that track
the indices of the circuit near which further optimization is needed.

We say that two fingers are \defn{non-interfering} if there are at least
$2\Omega$ gates between them. This property allows us to optimize the segments around these fingers in parallel without conflicts.


\begin{algorithm}[ht]
	\SetAlgoLined
	\caption{\name{} Algorithm for parallel optimization of quantum
		circuits.  The algorithm takes as input 1) an oracle optimizer denoted by ``$\mathrm{oracle}$'', 2) a gate array $\mathcal{A}$, and 3) a segment size
		$\Omega$, and outputs an optimized circuit.}
	\label{alg:poac}
	\SetKwBlock{Fn}{def POPQC($\mathrm{oracle}$: \(\textbf{gate array} \rightarrow \textbf{gate array}\),
		$\mathcal{A}$: \textbf{gate array},
		$\Omega$: \textbf{int}) $:$ $\textbf{gate array}$}{end}
	\Fn{
		\tcp{Initialize fingers}
		\(\mathcal{F} \gets \{\,0,\;\Omega,\;2 \cdot \Omega,\;\dots,\;\lfloor \frac{\abs{\mathcal{A}}}{\Omega} \rfloor \cdot \Omega \,\}\)

		\tcp{Create circuit}
		\(\mathcal{C} \gets \mathrm{Circuit}\textrm{.create}(\mathcal{A})\)

		\tcp{Optimize in rounds}
		\While{$\mathcal{F} \neq \emptyset$}{
			\(\mathcal{F} \gets \mathrm{optimizeSegments}(\mathcal{C}, \, \mathcal{F},\, \mathrm{oracle},\, \Omega)\)
		}

		\tcp{Return the gates of the optimized circuit}
		\Return{$\mathrm{Circuit.gates}(\mathcal{C})$}
	}
\end{algorithm}

\begin{algorithm}[ht]
	\SetAlgoLined
	\SetKwFor{ParMap}{parmap}{}{end}
	\caption{\textbf{optimizeSegments} \textbf{Algorithm}}
	\label{alg:optimize_segments}

	\SetKwBlock{Fn}{def optimizeSegments($\mathcal{C}$: circuit, $\mathcal{F}$: $\textbf{int array}$, $\mathrm{oracle}$: \(\textbf{gate array} \rightarrow \textbf{gate array}\), $\Omega$: \textbf{int}) $:$ $\textbf{int array}$}{end}
	\Fn{
		\tcp{Select non-interfering fingers to optimize}
		$\mathcal{F}_{selected},\, \mathcal{F}_{remaining} \gets \mathrm{selectFingers}(\mathcal{F},\, \mathcal{C},\, \Omega)$

		\tcp{Optimize segments around selected fingers independently}
		$(\mathcal{F}_{new},\, \mathcal{C}_{updates}) \gets$ \ParMap{\(f \in \mathcal{F}_{selected}\)}{
			$segment \gets \{\mathrm{Circuit.get}(\mathcal{C},\, \mathrm{Circuit.before}(\mathcal{C},\, f) + i) \;\mathrm{for}\; i \in [- \Omega, \Omega]\}$

			$optSegment \gets \mathrm{oracle}(segment)$

			\uIf{$\abs{{optSegment}} < \abs{{segment}}$}{
				\tcp{Update fingers}
				$\mathcal{F}_{new} \gets  \{\mathrm{Circuit.before}(\mathcal{C},\, f) - \Omega, \mathrm{Circuit.before}(\mathcal{C},\, f) + \Omega\}$

				$optSegment \gets \mathrm{padWithTombstone}(optSegment, |segment|)$

				\tcp{Collect updates to circuit}
				$\mathcal{C}_{updates} \gets\{(\mathrm{Circuit.before}(\mathcal{C},\, f) + i,\; optSegment[i+ \Omega])\;\mathrm{for}\; i \in [-\Omega, \Omega]\}$

				\Return{$(\mathcal{F}_{new},\, \mathcal{C}_{updates})$}
			}
			\Else{
				\Return{$(\emptyset, \emptyset)$}
			}
		}
		$\mathcal{F}_{new} \gets \bigcup_{f \in \mathcal{F}_{new}} f$

		$\mathcal{C}_{updates} \gets \bigcup_{c \in \mathcal{C}_{updates}} c$

		\tcp{Apply updates to circuit}
		$\mathrm{Circuit.substitute}(\mathcal{C},\, \mathcal{C}_{updates})$

		\tcp{Merge two sorted lists, removing duplicate items and maintaining sorted order}
		$\mathcal{F} \gets \mathrm{mergeAndDeduplicate}(\mathcal{F}_{remaining},\, \mathcal{F}_{new})$

		\Return{$\mathcal{F}$
		}
	}
\end{algorithm}

\begin{algorithm}[ht]
	\SetAlgoLined
	\caption{\textbf{selectFingers} \textbf{Algorithm}}
	\label{alg:partition_fingers}
	\SetKwFor{ParMap}{parmap}{}{end}
	\SetKwBlock{Fn}{def selectFingers($\mathcal{F}$: $\textbf{int array}$, $\mathcal{C}$: circuit, $\Omega$: \textbf{int}) : (\textbf{int array}, \textbf{int array})}{end}
	\Fn{
		$(\mathcal{F}_{even},\, \mathcal{F}_{odd}) \gets$ \ParMap{\(i \in \{0,\cdots,\abs{\mathcal{F}}\}\)}{
			$groupIndex \gets \lfloor \mathrm{Circuit.before}(\mathcal{C},\, \mathcal{F}[i]) / 2\Omega\rfloor$

			$groupIndexPrev \gets \textbf{if}\; i > 0 \;\textbf{then}\; \lfloor \mathrm{Circuit.before}(\mathcal{C},\, \mathcal{F}[i-1]) / 2\Omega\rfloor \;\textbf{else}\; -1$

			\tcp{Determine if $f$ is the first finger in its group by comparing its group index and the group index of the previous finger}

			\uIf{$groupIndex > groupIndexPrev $}{

				\uIf{$groupIndex \bmod 2 = 0$}{
					\Return{$(\{i\}, \emptyset)$}
				}
				\Else{
					\Return{$(\emptyset, \{i\})$}
				}
			}
			\Else{
				\Return{$(\emptyset, \emptyset)$}
			}

		}

		$\mathcal{F}_{even} \gets \bigcup_{f \in \mathcal{F}_{even}} f$

		$\mathcal{F}_{odd} \gets \bigcup_{f \in \mathcal{F}_{odd}} f$

		\tcp{Return the larger of the two sets}

		\uIf{$\abs{\mathcal{F}_{even}} > \abs{\mathcal{F}_{odd}}$}{
			\Return{$ \mathcal{F}_{even},\mathcal{F} \setminus \mathcal{F}_{even}$}
		}
		\Else{
			\Return{$ \mathcal{F}_{odd},\mathcal{F} \setminus \mathcal{F}_{odd}$}
		}
	}
\end{algorithm}

Algorithm~\ref{alg:poac} shows the pseudo-code for the \name{} algorithm.
The algorithm is parameterized by a \emph{local} oracle optimizer that can
optimize a small circuit segment.
We represent the oracle as a function and make no assumptions about its inner
workings.
In addition to the oracle, the algorithm takes as input the gate array
$\mathcal{A}$ representing the circuit to optimize, and a segment size
$\Omega$.

The algorithm maintains a set of fingers ($\mathcal{F}$) and optimizes segments
around them.
More precisely, given a finger at some position in the circuit, the algorithm
presumes that
any $\Omega$-segment \defn{containing} the finger needs to be optimized.
We say that a segment from index $i$ to $j$ contains a finger at $f$, if $i \leq
	f < j$.

%
%

To optimize the input circuit, the algorithm proceeds in rounds. Each round of
optimization uses the \textbf{optimizeSegments} algorithm, which
starts by partitioning the fingers into two sets: $\mathcal{F}_{selected}$,
which contains non-interfering fingers that will be optimized, and
$\mathcal{F}_{remaining}$, which contains the remaining fingers that will not
be optimized in the current round.
This is done using the \textbf{selectFingers} algorithm
(Algorithm~\ref{alg:partition_fingers}).
Then, for each finger $f \in \mathcal{F}_{selected}$,
the \textbf{optimizeSegments} algorithm finds a $2\Omega$-segment centered at
$f$ using the index tree and optimizes the segment using the oracle.
Because the selected fingers are non-interfering, the algorithm can optimize
around each finger in parallel.
As the algorithm optimizes the segments around each selected finger, it creates
a new set of fingers $\mathcal{F}_{new}$ around the segments that are optimized
by the oracle.
Specifically, if the oracle does not optimize the segment (determined by
comparing the size or the gate count of the optimized segment to the
unoptimized one, denoted as $\abs{segment}$), then the algorithm removes the
finger.
Otherwise, it adds new fingers at the boundaries of the optimized segment. As
we will show in Lemma~\ref{lem:optimize_fingers_correctness}, the
\textbf{optimizeSegments} algorithm preserves the invariant that each
unoptimized $\Omega$-segment contains a finger.
%

The \textbf{selectFingers} algorithm partitions the set of fingers into two
sets, where
the first set contains non-interfering fingers and the second set contains the
remaining fingers, while trying to maximize the size of the first set.
To do so, the \textbf{selectFingers} algorithm first partitions the circuit
into groups of $2\Omega$ gates each, except possibly for the last group.
It then selects the first finger from each even-numbered group to construct the
set $\mathcal{F}_{even}$.
Similarly, the algorithm selects the first finger from each odd-numbered group
to construct the set $\mathcal{F}_{odd}$.
%
%
The algorithm then selects the larger of $\mathcal{F}_{odd}$ and
$\mathcal{F}_{even}$, and returns the corresponding partition of the fingers to
ensure maximum progress. Actually, it is guaranteed that this
\textbf{selectFingers} algorithm selects a constant fraction of the fingers as
we will show in Lemma~\ref{lem:efficiency_of_partition_fingers}.

\begin{figure}[t]
	\centering
	\includegraphics[width=\columnwidth]{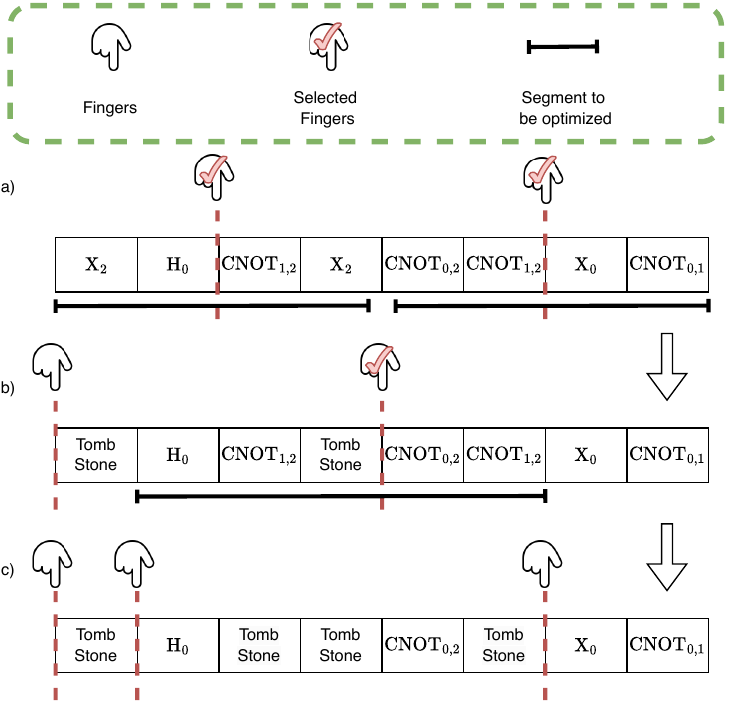}
	\caption{An illustration of a run of the \name{} algorithm. We
		assume $\Omega =2$, and thus the segments being optimized
		consist of $2\Omega = 4$ gates.}
	\label{fig:example}
\end{figure}

Figure~\ref{fig:example} illustrates an example run of the \name{} algorithm.
As shown in Figure~\ref{fig:example}a, we initially have two fingers at indices
2 and 6.
These two fingers are non-interfering, and thus, in the first round, the
\textbf{optimizeSegments} algorithm optimizes the segments centered at these
fingers in parallel.
This optimization removes the two \texttt{X} gates from the left segment, 
while the right segment remains unaffected (no optimizations are needed).
The algorithm then removes the finger within the second, unaffected segment and
adds the new fingers at the boundaries of the first, optimized segment, as
shown in Figure~\ref{fig:example}b.
In the second round, the two fingers interfere, so only one of them is selected
for optimization.
Assuming the right finger is selected, the algorithm optimizes the
segment containing it, which deletes two \texttt{CNOT} gates as shown in
Figure~\ref{fig:example}c.
This algorithm continues optimization in rounds until no more fingers remain.

\section{Efficiency Analysis}
In this section, we analyze the termination and efficiency of the \name{}
algorithm. Our analysis relies on two key lemmas that characterize the
efficiency of finger selection and the total number of oracle calls required.

\begin{lemma}[Efficiency of \textbf{selectFingers}]
	\label{lem:efficiency_of_partition_fingers}
	The \textbf{selectFingers} algorithm has work $O(\abs{\mathcal{F}}\lg{n})$ and span $O(\lg{n})$.  Furthermore, the returned set of non-interfering fingers contains at least a $\frac{1}{4\Omega}$ fraction of all fingers.
\end{lemma}
\begin{proof}
	The \textbf{selectFingers} algorithm processes all fingers in parallel. For each finger $f \in \mathcal{F}$, the algorithm performs two \textbf{before} operations involving the index tree, which take $O(\lg{n})$ time. Therefore, the work is $O(\abs{\mathcal{F}}\lg{n})$ and the span is $O(\lg{n})$.

	To prove that the number of selected fingers is a constant fraction of all the
	fingers, note that
	\[
		\abs{\mathcal{F}_{even}} +
		\abs{\mathcal{F}_{odd}} \geq \frac{\abs{\mathcal{F}}}{2\Omega},
	\]
	because each group has at most $2\Omega$ fingers and for each group with at
	least one finger, one of the fingers is selected.
	Thus, we have $\max(\abs{\mathcal{F}_{even}}, \abs{\mathcal{F}_{odd}}) \geq
		\frac{\abs{\mathcal{F}}}{4\Omega}$, which means at least a $\frac{1}{4\Omega}$
	fraction of all fingers are selected.
\end{proof}

\begin{lemma}[Oracle Call Efficiency]
	\label{lem:oracle_call_efficiency}
	The \name{} algorithm requires $O(n)$ oracle calls.
\end{lemma}
\begin{proof}
	We prove this lemma by defining a potential function $L = \abs{\mathcal{F}} + 2\abs{\mathcal{C}}$, then showing that it is bounded by $O(n)$ and decreases by at least $1$ for each oracle call.
	
	For each oracle call, there are two cases:
	\begin{itemize}
			      \item[a)] The oracle does not optimize the segment. Only the selected finger is
			            removed, and $\abs{\mathcal{F}}$ decreases by $1$.
			      \item[b)] The oracle makes changes to the circuit. This means that the size of the
			            optimized circuit is reduced by at least $1$, so $\abs{\mathcal{C}}$ decreases
			            by at least $1$. The original finger is removed and two new fingers are added,
			            so $\abs{\mathcal{F}}$ increases by $1$.
		\end{itemize}
		Therefore, $L$ decreases by at least $1$ for each oracle call in either case.
		Initially, we have $L = \ceil{\frac{n}{\Omega}} + 2n$.
		Hence, the number of oracle calls is $O(n)$.
\end{proof}

We have the following lemma by combining the two lemmas above.

\begin{lemma}[Number of Fingers]
	\label{lem:fingers_summation}
	Let $\mathcal{F}^i$ be the set of fingers in the $i$-th round. The total number of fingers across all rounds, $\sum_i \abs{\mathcal{F}^i}$, is $O(\Omega n)$.
\end{lemma}

\begin{proof}
	From Lemma~\ref{lem:efficiency_of_partition_fingers}, we know that in the $i$-th round, the number of selected fingers satisfies $\abs{\mathcal{F}_{selected}^i} \geq \frac{\abs{\mathcal{F}^i}}{4\Omega}$.
	By Lemma~\ref{lem:oracle_call_efficiency}, the total number of oracle calls
	across all rounds is $O(n)$. Since each selected finger results in exactly one
	oracle call, we have $\sum_i \abs{\mathcal{F}_{selected}^i} = O(n)$. Therefore,
	$\sum_i \abs{\mathcal{F}^i} \leq 4\Omega \cdot \sum_i
		\abs{\mathcal{F}_{selected}^i} = 4\Omega \cdot O(n) = O(\Omega n)$.
\end{proof}

Finally, we prove the work and span bounds on the \name{} algorithm.

\begin{theorem}[Work and Span of \name{}]
	\label{thm:poac_efficiency}
	Suppose the upper bound of the work of an oracle call on a $2\Omega$-segment is $W$ and the span is $S$. The \name{} algorithm has almost linear work $O(n(\Omega\lg{n}+W))$ in total and logarithmic span $O(r(\lg{n}+S))$, where $r$ is the number of rounds (iterations of the outer loop) and $n$ is the size of the circuit.
\end{theorem}

\begin{proof}
	We analyze the work and span of the \textbf{optimizeSegments} function by examining each component:
	\begin{itemize}
		\item \textbf{selectFingers}: As established in Lemma~\ref{lem:efficiency_of_partition_fingers}, this has work $O(\abs{\mathcal{F}}\lg{n})$ and span $O(\lg{n})$.

		\item \textbf{Parallel optimization}: For each finger $f \in \mathcal{F}_{selected}$, we perform:
		      \begin{itemize}
			      \item Segment extraction: $O(\Omega\lg{n})$ work and $O(\lg{n})$ span
			      \item Oracle call: $O(W)$ work and $O(S)$ span
			      \item Collecting updates: $O(\Omega)$ work and $O(\lg{\Omega})$ span
		      \end{itemize}
		      This gives a total work of $O(\abs{\mathcal{F}_{selected}}(\Omega\lg{n} + W))$ and span $O(\lg{n} + S)$ for this phase. Here, we use the fact that $\Omega < n$.

		\item \textbf{Circuit substitution}: The \textbf{substitute} function processes input of length $O(\Omega \cdot \abs{\mathcal{F}_{selected}})$, requiring $O(\Omega \cdot \abs{\mathcal{F}_{selected}} \cdot \lg{n})$ work with $O(\lg{n})$ span.

		\item \textbf{Finger merging}: The \textbf{mergeAndDeduplicate} operation has work $O(\abs{\mathcal{F}})$ and span $O(\lg \abs{\mathcal{F}})$.
	\end{itemize}

	For each round, the total work is $O(\abs{\mathcal{F}}\lg{n} +
		\abs{\mathcal{F}_{selected}}(\Omega\lg{n} + W))$ and the span per round is
	$O(\lg{n} + S)$.

	By Lemma~\ref{lem:fingers_summation}, the sum of $\abs{\mathcal{F}}$ across all
	rounds is $O(\Omega n)$. By Lemma~\ref{lem:oracle_call_efficiency}, the sum of
	$\abs{\mathcal{F}_{selected}}$ across all rounds is $O(n)$. Therefore, the
	total work is $O(n(\Omega\lg{n} + W))$, and the span is $O(r(\lg{n} + S))$, where
	$r$ is the number of rounds. While $r$ could theoretically be as large as
	$\Theta(n)$ in the worst case, our empirical results in the next section
	demonstrate that $r$ is typically a small constant in practice.
\end{proof}

In practice, we choose $\Omega$ to be a constant factor that depends on the
oracle.
As we increase $\Omega$, the oracle and thus the run-time will increase, but
the quality may not increase due to natural locality of circuits.
The goal will be to find a setting for $\Omega$ that delivers good quality at
reasonable cost.

Assuming that $\Omega$ is constant, we can conclude that the work and span of
the oracle is also constant.
As a corollary to Theorem~\ref{thm:poac_efficiency}, we therefore conclude that
the work of the \name{} algorithm is $O(n\lg{n})$ and the span is $O(r\lg{n})$.

\section{Local Optimality of \name{}}
We define local optimality with respect to a given oracle function
``$\mathrm{oracle}$'' and segment size $\Omega$ as follows: for a circuit
represented by a gate array $\mathcal{A}$, we say $\mathcal{A}$ is
\defn{locally optimal} if for any $\Omega$-segment $\mathcal{A}'$,
applying the oracle optimizer does not reduce its size, i.e.,
$\abs{\mathcal{A}'} \leq \abs{\mathrm{oracle}(\mathcal{A}')}$.

For this definition to be meaningful, we require the oracle function to be well-behaved.
An oracle is \defn{well-behaved} if, after it has optimized a circuit, any segment of its output is optimal with respect to the oracle. Formally, if
$\mathcal{A}' = \mathrm{oracle}(\mathcal{A})$, then for any
segment $\mathcal{A}''$ of $\mathcal{A}'$, we have $\abs{\mathcal{A}''} \leq
	\abs{\mathrm{oracle}(\mathcal{A}'')}$.

This property ensures that when we call the oracle with a $2\Omega$-segment, the output is locally optimal with respect to the oracle and segment size $\Omega$.

Before proving the local optimality of \name{}, we first establish two correctness
lemmas for the \textbf{selectFingers} and \textbf{optimizeSegments}
algorithms.
\begin{lemma}[Correctness of \textbf{selectFingers}]
	The \textbf{selectFingers} algorithm returns a partition of the finger set such that the first part contains non-interfering fingers (for any two selected fingers, there are at least $2\Omega$ gates between them).
	\label{lem:partition_fingers_correctness}
\end{lemma}

\begin{proof}
	The algorithm first computes two groups of fingers, even-numbered
	($\mathcal{F}_{even}$) and odd-numbered ($\mathcal{F}_{odd}$).
	Because it selects only one finger from each group, any pair of fingers in
	$\mathcal{F}_{even}$ are from different groups and there is always an odd group
	between them.
	Therefore, any pair of fingers in $\mathcal{F}_{even}$ is separated by at least a $2\Omega$-segment, and is non-interfering.
	Similarly, any pair of fingers in $\mathcal{F}_{odd}$ is also non-interfering.
\end{proof}

We now prove the correctness of \textbf{optimizeSegments} by showing that it
preserves the invariant that every unoptimized $\Omega$-segment contains a finger.

\begin{lemma}[Correctness of \textbf{optimizeSegments}]
	\label{lem:optimize_fingers_correctness}
	For any well-behaved oracle, if the input circuit $\mathcal{C}$ and
	the fingers $\mathcal{F}$ satisfy the invariant that every unoptimized
	$\Omega$-segment contains a finger, then \textbf{optimizeSegments}
	returns an updated circuit $\mathcal{C}'$ and updated fingers
	$\mathcal{F}'$ that maintain this invariant: each unoptimized
	$\Omega$-segment in $\mathcal{C}'$ contains a finger from $\mathcal{F}'$. Besides, the updated fingers $\mathcal{F}'$ are sorted.
\end{lemma}

\begin{proof}
	During the execution of \textbf{optimizeSegments}, for each selected finger $f\in\mathcal{F}_{selected}$, one of the following two cases occurs:
	\begin{enumerate}
		\item The oracle makes no changes to the $2\Omega$-length segment centered at $f$.
		      This implies that all $\Omega$-length subsegments within this region are locally
		      optimal by the well-behaved property of the oracle. The finger $f$ can be
		      safely removed.

		\item The oracle optimizes the $2\Omega$-length segment centered at $f$. After
		      optimization, any new $\Omega$-length segments fall into two categories: 1)
		      Segments fully contained within the optimized region, which are optimal by the
		      well-behaved property of the oracle. 2) Segments that cross the boundaries of
		      the optimized region, which contain the boundary fingers we placed at the start
		      and end of the optimized region. Therefore, all potentially unoptimized
		      $\Omega$-length segments are properly tracked with the updated fingers.
	\end{enumerate}
	Thus, the invariant is preserved: each unoptimized $\Omega$-segment in the updated circuit $\mathcal{C}'$ contains a finger from the updated set $\mathcal{F}'$. 
	
	We note that the set $\mathcal{F}_{new}$ is sorted because the selected fingers are non-interfering and $\mathcal{F}_{remaining}$ is also sorted. So a simple merge operation is enough to maintain the sorted property.
\end{proof}

Based on these lemmas, we now prove that the function \name{} is correct.
\begin{theorem}[Local Optimality of \name{}]
	\label{thm:name_correctness}
	The \name{} algorithm produces a locally optimal circuit with respect to a given oracle function $\mathrm{oracle}$ and segment size $\Omega$.
\end{theorem}

\begin{proof}
	The algorithm starts by initializing the fingers $\mathcal{F} = [0,
		\Omega, 2\Omega, \cdots]$, placing a finger at the start of
	each $\Omega$-segment.
	This ensures our initial invariant holds: every unoptimized $\Omega$-segment
	contains a finger. We note that the choice of the initial fingers is not unique, but this specific choice minimizes the initial number of fingers.
	By Lemma~\ref{lem:optimize_fingers_correctness}, each call to
	\textbf{optimizeSegments} preserves the invariant that every unoptimized
	$\Omega$-segment contains a finger.
	Furthermore, by Lemma~\ref{lem:efficiency_of_partition_fingers}, in each round,
	at least one finger is selected, and by Lemma~\ref{lem:oracle_call_efficiency},
	the total number of oracle calls is bounded by $O(n)$. The algorithm
	terminates in at most $O(n)$ rounds.

	Since our invariant guarantees that every unoptimized $\Omega$-segment contains a finger, the absence of fingers implies that no unoptimized $\Omega$-segments
	remain in the circuit. Therefore, the final circuit is locally optimal with
	respect to the oracle and segment size $\Omega$.
\end{proof}

\section{Implementation and Evaluation}
\label{sec:evaluation}

Thus far in the paper, we have presented an algorithm that guarantees
$O(n\lg{n})$ work and $O(r\lg{n})$ span (assuming that $\Omega$ is a constant).
Is this algorithm practical and does it perform well in practice? Specifically,
can it take advantage of parallelism effectively to optimize large circuits?
In this section, we describe an implementation and present an evaluation that
answers these questions affirmatively.

\subsection{Implementation}
We implement the described algorithm in the Rust language, using the Rayon
library which provides fork-join parallelism.
~\footnote{We implemented an earlier version of the algorithm in Parallel ML~\cite{westrick+disent-2020,awa+space-2021,awa+ent-2023} but
later changed to Rust to make the implementation more accessible.}
Our implementation matches closely with the algorithm description.
%
%
%
In our implementation, we did not attempt to optimize manually the
overheads of parallelism (e.g., via granularity
control~data~\cite{acarchra11,acg-heartbeat-2018})
but instead left it to Rust+Rayon's adaptive loop-splitting strategy,
which attempts to avoid the overheads of parallelism when the loop
iterations are computationally insignificant.
As we discuss, this works well in most cases, except perhaps for the smallest
circuits, when parallelism overheads can be proportionally more significant.

In our implementation, we primarily use the \voqc{}\cite{hietala2021verified}
optimizer, because at the circuit sizes we consider,
it is the fastest optimizer and produces the best quality outputs within a
reasonable amount of time (e.g., 24 hours).
(See Section~\ref{sec:related} for a discussion of other optimizers.)
Our implementation calls the oracle via a system call and is therefore
relatively easy to adapt to use other optimizers as oracles.
As an additional oracle, we use the \quartz{}\cite{quartz-2022} optimizer,
which is significantly slower than \voqc{} but allows us to optimize other cost
metrics, such as the depth of the circuit.
We note that both oracle optimizers are fast for small to moderate size
circuits consisting of several thousand gates, but perform poorly on larger
circuits, which are necessary to realize the benefits of quantum computing.
%

%

For our experiments, we choose $\Omega=200$ because it provides a good balance
between speed and quality.
The results are not particularly sensitive to the exact setting of $\Omega$
within the range from $100$ to $800$ and in principle, $\Omega$ can be set to a circuit-dependent value.
\iffull 
(Section~\ref{sec:exp::omega}).
\else
(see extended version\cite{liu2025popqc} for details).
\fi

\subsection{Benchmarks}
To evaluate the performance and effectiveness of our techniques, we select
benchmarks from existing benchmarking suites, including
PennyLane\cite{bergholm2018pennylane}, Qiskit\cite{qiskit-2019}, and
NWQBench\cite{li2021qasmbench}.
The benchmarks include boolean satisfaction problems (\texttt{BoolSat}), the
Binary Welded Tree (\texttt{BWT}), Grover's searching algorithm
(\texttt{Grover})\cite{grover1996fast}, the HHL algorithm for solving linear
equations (\texttt{HHL})\cite{harrow2009quantum}, Shor's algorithm for
factoring integers (\texttt{Shor})\cite{shor1994algorithms}, square-root
algorithm (\texttt{Sqrt}), state vector preparation (\texttt{StateVec}), and
Variational Quantum Eigensolver (\texttt{VQE})\cite{peruzzo2014variational}.
We choose the specific benchmarks because their size scales rapidly with qubit
counts, and they remain challenging for state-of-the-art optimizers.
Our benchmarks and oracles all use a gate set consisting of Hadamard
(\texttt{H}), Pauli-X (\texttt{X}), controlled-not (\texttt{CNOT}), and
rotation-Z (\texttt{RZ}), which is the gate set used by \voqc{}.
%


\subsection{Evaluation Setup}
We conduct experiments on a machine with 64 cores (AMD EPYC 7763) and 256GB
RAM.
To evaluate our optimizer \name{}, we first compare it to the \voqc{} optimizer
(Section~\ref{sec:exp::voqc}).
Because \voqc{} is sequential, this comparison includes two advantages of our
optimizer, locality and parallelism.
We then separately analyze the impact of local optimality
(Section~\ref{sec:exp::local}) and parallelism (Section~\ref{sec:exp::par}) on
the overall performance.

Subsequently, we examine the work efficiency of our optimizer \iffull
	(Sections~\ref{sec:exp::work} and \ref{sec:exp::oracletime}), \else
	(Section~\ref{sec:exp::work} and extended version\cite{liu2025popqc} for details), \fi demonstrate its flexibility by
integrating \quartz{} as an alternative oracle optimizer with a layered circuit
representation (Section~\ref{sec:exp::quartz}), and analyze the sensitivity of
our method to different values of $\Omega$
\iffull
	(Section~\ref{sec:exp::omega}).
\else
	(see extended version\cite{liu2025popqc} for details).
\fi

\subsection{Our Optimizer is Orders of Magnitude Faster than VOQC}
\label{sec:exp::voqc}
\begin{table*}[ht]
	\begin{tabular*}{\textwidth}{@{\extracolsep{\fill}} cccccccc}
		                                   &            &           & \multicolumn{2}{c}{VOQC(1 thread)}                      & \multicolumn{3}{c}{\name{}(64 threads)}                                                                                    \\           \cmidrule(lr){4-5} \cmidrule(lr){6-8}

		benchmark                          & $\#$qubits & $\#$gates & \begin{tabular}[c]{@{}c@{}}gate\\reduction\end{tabular} & time(s)                                 & \begin{tabular}[c]{@{}c@{}}gate\\reduction\end{tabular} & time(s) & speedup      \\ \midrule
		\multirow{4}{*}{\texttt{BoolSat}}  & 28         & 75818     & 83.2\%                                                  & 145.5                                   & 83.7\%                                                  & 3.6     & 40.2         \\
		                                   & 30         & 138443    & 83.3\%                                                  & 722.5                                   & 83.6\%                                                  & 4.6     & 155.8        \\
		                                   & 32         & 262724    & 83.3\%                                                  & 3055.0                                  & 83.4\%                                                  & 5.6     & 544.2        \\
		                                   & 34         & 510137    & 83.3\%                                                  & 15952.6                                 & 83.3\%                                                  & 7.6     & 2091.1       \\
		\midrule
		\multirow{4}{*}{\texttt{BWT}}      & 17         & 361603    & 44.7\%                                                  & 12165.6                                 & 44.7\%                                                  & 6.9     & 1770.0       \\
		                                   & 21         & 553603    & 51.4\%                                                  & 32549.3                                 & 51.4\%                                                  & 12.0    & 2712.5       \\
		                                   & 25         & 946801    & N.A.                                                    & $\geq$86400.0                           & 52.9\%                                                  & 17.2    & $\geq$5027.9 \\
		                                   & 29         & 1298801   & N.A.                                                    & $\geq$86400.0                           & 53.9\%                                                  & 20.0    & $\geq$4326.3 \\
		\midrule
		\multirow{4}{*}{\texttt{Grover}}   & 9          & 8968      & 29.4\%                                                  & 5.8                                     & 29.3\%                                                  & 1.1     & 5.1          \\
		                                   & 11         & 27136     & 29.9\%                                                  & 63.8                                    & 29.6\%                                                  & 1.5     & 42.9         \\
		                                   & 13         & 72646     & 29.7\%                                                  & 565.3                                   & 29.3\%                                                  & 2.1     & 264.2        \\
		                                   & 15         & 180497    & 29.5\%                                                  & 3911.3                                  & 28.9\%                                                  & 3.4     & 1151.2       \\
		\midrule
		\multirow{4}{*}{\texttt{HHL}}      & 7          & 5796      & 44.5\%                                                  & 0.3                                     & 58.9\%                                                  & 0.8     & 0.4          \\
		                                   & 9          & 68558     & 44.7\%                                                  & 151.1                                   & 59.5\%                                                  & 1.7     & 89.2         \\
		                                   & 11         & 680376    & 41.9\%                                                  & 33483.9                                 & 56.5\%                                                  & 6.0     & 5600.6       \\
		                                   & 13         & 5954308   & N.A.                                                    & $\geq$86400.0                           & 55.9\%                                                  & 35.1    & $\geq$2464.3 \\
		\midrule
		\multirow{4}{*}{\texttt{Shor}}     & 10         & 8476      & 11.1\%                                                  & 5.4                                     & 10.9\%                                                  & 0.5     & 10.7         \\
		                                   & 12         & 31267     & 3.2\%                                                   & 106.6                                   & 3.2\%                                                   & 0.6     & 173.5        \\
		                                   & 14         & 136320    & 11.3\%                                                  & 2276.9                                  & 11.1\%                                                  & 1.6     & 1451.7       \\
		                                   & 16         & 545008    & 11.3\%                                                  & 53486.1                                 & 11.1\%                                                  & 4.1     & 13110.7      \\
		\midrule
		\multirow{4}{*}{\texttt{Sqrt}}     & 42         & 111956    & 42.2\%                                                  & 442.8                                   & 41.3\%                                                  & 3.4     & 132.0        \\
		                                   & 48         & 258725    & 42.2\%                                                  & 3154.9                                  & 40.0\%                                                  & 5.4     & 585.3        \\
		                                   & 54         & 585234    & 42.2\%                                                  & 17854.0                                 & 38.8\%                                                  & 9.5     & 1875.2       \\
		                                   & 60         & 1306507   & N.A.                                                    & $\geq$86400.0                           & 37.9\%                                                  & 17.7    & $\geq$4879.6 \\
		\midrule
		\multirow{4}{*}{\texttt{StateVec}} & 5          & 32147     & 79.6\%                                                  & 12.9                                    & 79.6\%                                                  & 1.5     & 8.7          \\
		                                   & 6          & 134632    & 79.2\%                                                  & 605.1                                   & 79.1\%                                                  & 2.6     & 230.4        \\
		                                   & 7          & 546035    & 78.9\%                                                  & 15272.9                                 & 78.8\%                                                  & 3.7     & 4084.1       \\
		                                   & 8          & 2175747   & N.A.                                                    & $\geq$86400.0                           & 78.7\%                                                  & 9.6     & $\geq$9027.5 \\
		\midrule
		\multirow{4}{*}{\texttt{VQE}}      & 18         & 29800     & 64.4\%                                                  & 8.8                                     & 64.8\%                                                  & 0.9     & 9.8          \\
		                                   & 22         & 48448     & 61.6\%                                                  & 37.3                                    & 62.0\%                                                  & 1.1     & 33.6         \\
		                                   & 26         & 72600     & 59.0\%                                                  & 122.7                                   & 59.3\%                                                  & 1.4     & 85.3         \\
		                                   & 30         & 102768    & 56.5\%                                                  & 308.8                                   & 56.9\%                                                  & 1.4     & 228.5        \\
		\midrule
		average                                &            &           & 48.94\%                                                 &                                         & 51.20\%                                                 &         & $\geq$1944.14      \\\bottomrule
	\end{tabular*}
	\caption{Optimization quality (represented as gate reduction) and running time comparison of \name{} and \voqc{} on a set of quantum benchmarks. \voqc{} is sequential and \name{} is executed on 64 threads.}
	\label{tab:main_result}
\end{table*}

We compare the speed and quality of our \name{} optimizer with the \voqc{}
optimizer, which is sequential (Table~\ref{tab:main_result}).
We note that in some cases, our baseline optimizer \voqc{} does not terminate
in our timeout of 24 hours.
We indicate non-terminating runs with an ``N.A.'' in
Table~\ref{tab:main_result} and exclude their missing data from gate reduction averages.
Given the large timeout, it may seem surprising that optimizers can take such a
long time, but it is common indeed in quantum circuit optimizers.
Such large run times are one of the motivations behind this work.
Our parallel optimizer terminates on all benchmarks and does so relatively
quickly, achieving more than $10^3\times$ speedup on average across all
benchmarks.
As the table shows, the speedups increase as the circuit sizes increase, e.g.,
for \texttt{Shor} with 16 qubits, the speedup is more than $10^4\times$.
We also observe that the performance improvements come at little to no
deterioration in optimization quality: for most benchmarks, we see a slight
decrease in quality, e.g., $0.5\%$, while in one benchmark \texttt{HHL}, our
optimizer \name{} improves quality by more than $10\%$ over \voqc{}.

It's not intuitive why \name{} outperforms its base oracle \voqc{} in some cases. This occurs because \voqc{} applies optimization passes sequentially. A later pass might create opportunities for an earlier pass, which \voqc{} would miss in a single execution. In contrast, \name{} applies the oracle multiple times to nearby segments and can capture these opportunities, effectively behaving like ``running \voqc{} until convergence'' rather than a single pass.
%
%
%


The fact that we do not see significant degradation in quality may come across
as surprising, because our optimizer only performs local optimization.
We attribute this outcome to two factors.
First, our locality requirement applies at each and every $\Omega$-segment; it
therefore is a reasonably strong guarantee in practice, especially for moderate
values of $\Omega$ (e.g., $\Omega \ge 100$).
Second, much like classical algorithms, quantum circuits naturally possess some
degree of locality, because each segment of the circuit performs a specific
function, leading to few, if any, optimizations across distant gates.
%


\subsection{Local Optimality Unlocks Significant Efficiency}
\label{sec:exp::local}
\begin{table*}[ht]
	\begin{tabular*}{\textwidth}{@{\extracolsep{\fill}} cccccc}
		benchmark                          & $\#$qubits & VOQC(1 thread) time & \name{}(1 thread) time & speedup \\
		\midrule
		\multirow{4}{*}{\texttt{BoolSat}}  & 28         & 145.47              & 20.76                  & 7.0     \\
		                                   & 30         & 722.50              & 38.25                  & 18.9    \\
		                                   & 32         & 3055.01             & 72.32                  & 42.2    \\
		                                   & 34         & 15952.59            & 142.03                 & 112.3   \\
		\midrule
		\multirow{4}{*}{\texttt{BWT}}      & 17         & 12165.58            & 164.48                 & 74.0    \\
		                                   & 21         & 32549.34            & 325.93                 & 99.9    \\
		                                   & 25         & $\geq$86400.00            & 527.38                 & $\geq$163.8   \\
		                                   & 29         & $\geq$86400.00            & 672.45                 & $\geq$128.5   \\
		\midrule
		\multirow{4}{*}{\texttt{Grover}}   & 9          & 5.78                & 3.52                   & 1.6     \\
		                                   & 11         & 63.76               & 10.97                  & 5.8     \\
		                                   & 13         & 565.35              & 31.22                  & 18.1    \\
		                                   & 15         & 3911.32             & 77.78                  & 50.3    \\
		\midrule
		\multirow{4}{*}{\texttt{HHL}}      & 7          & 0.32                & 1.23                   & 0.3     \\
		                                   & 9          & 151.12              & 14.38                  & 10.5    \\
		                                   & 11         & 33483.88            & 154.47                 & 216.8   \\
		                                   & 13         & $\geq$86400.00            & 1338.62                & $\geq$64.5    \\
		\midrule
		\multirow{4}{*}{\texttt{Shor}}     & 10         & 5.43                & 2.14                   & 2.5     \\
		                                   & 12         & 106.58              & 7.20                   & 14.8    \\
		                                   & 14         & 2276.87             & 34.07                  & 66.8    \\
		                                   & 16         & 53486.13            & 135.62                 & 394.4   \\
		\midrule
		\multirow{4}{*}{\texttt{Sqrt}}     & 42         & 442.84              & 53.56                  & 8.3     \\
		                                   & 48         & 3154.85             & 125.29                 & 25.2    \\
		                                   & 54         & 17854.00            & 280.18                 & 63.7    \\
		                                   & 60         & $\geq$86400.00            & 632.01                 & $\geq$136.7   \\
		\midrule
		\multirow{4}{*}{\texttt{StateVec}} & 5          & 12.95               & 4.47                   & 2.9     \\
		                                   & 6          & 605.14              & 18.90                  & 32.0    \\
		                                   & 7          & 15272.95            & 69.92                  & 218.4   \\
		                                   & 8          & $\geq$86400.00            & 272.07                 & $\geq$317.6   \\
		\midrule
		\multirow{4}{*}{\texttt{VQE}}      & 18         & 8.79                & 3.57                   & 2.5     \\
		                                   & 22         & 37.27               & 5.83                   & 6.4     \\
		                                   & 26         & 122.73              & 9.01                   & 13.6    \\
		                                   & 30         & 308.75              & 12.67                  & 24.4    \\
		\midrule
		average                                &            &                     &                        & $\geq$73.3    \\
		\bottomrule
	\end{tabular*}
	\caption{Running time comparison (in seconds) between \name{} and \voqc{}, both executed on a single thread. The speedup column shows how many times faster \name{} is compared to \voqc{}.}
	\label{tab:one_thread}
\end{table*}
To measure the efficiency improvement due to local optimality (compared to
global optimality), we compare the single-core run time of our \name{}
optimizer with the \voqc{} optimizer.
As can be seen in Table~\ref{tab:one_thread}, \name{} achieves more than
$70\times$ speedup on average.
This shows that significant speedup is due to local optimality.
We note that we do not include the output circuit sizes in this table, but they
are the same as in the parallel experiments (Table~\ref{tab:main_result}),
which show that (as discussed above) these improvements come without noticeable
degradation in quality of optimization.

\subsection{Our Optimizer Scales Well as Number of Cores Increases}
\label{sec:exp::par}

\begin{figure}[t]
	\centering
	\includegraphics[width=\columnwidth]{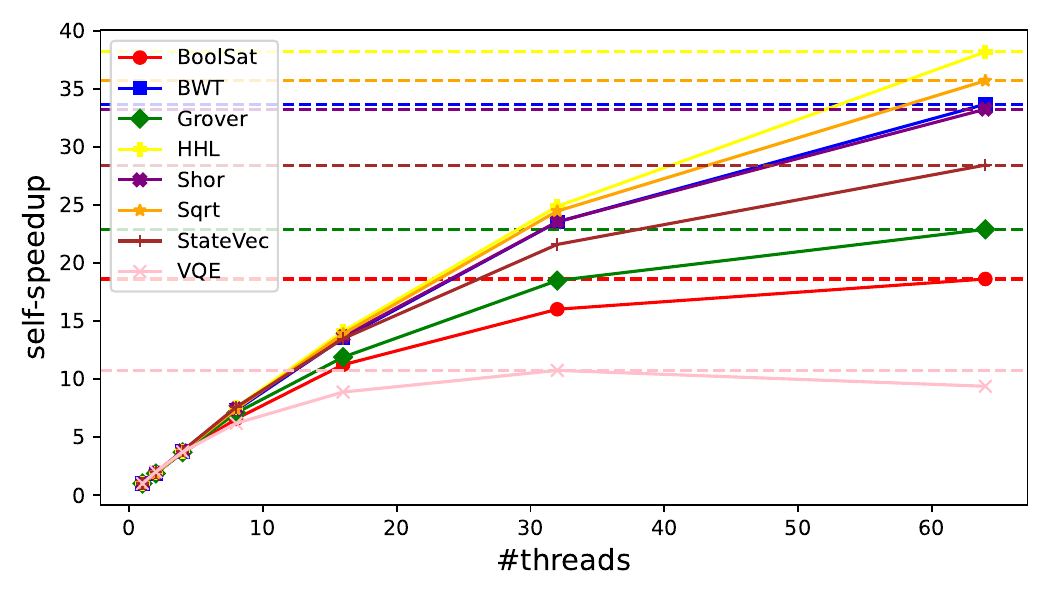}
	\caption{Self-speedup using different numbers of threads with respect to the single-thread case. The data points correspond to the largest instances in each circuit family.}
	\label{fig:speedup_vs_n_threads}
\end{figure}


To evaluate the scalability of our optimizer \name{} as we increase the number
of cores, we run our optimizer with different numbers of cores up to the 64
cores of our experiment machine and calculate the self-speedup with respect to
the single-core run.
We use self-speedups, rather than calculating speedups with respect to another
oracle, for two reasons. First, our oracle is the fastest (on a uniprocessor)
of all other sequential oracles that we have experimented with.
Second, for scalability analysis, self-speedups are more revealing as they
factor out other concerns, such as algorithmic and implementation differences.

Figure~\ref{fig:speedup_vs_n_threads} shows the speedups.
We can discern two patterns.
A majority of the benchmarks \texttt{HHL}, \texttt{Sqrt}, \texttt{BWT},
\texttt{Shor}, \texttt{StateVec} and \texttt{Grover} scale well, achieving
speedups of $20$ fold or more.
The remaining benchmarks, \texttt{VQE} and \texttt{BoolSat}, scale less well.
%


%
%

To understand the scalability of \texttt{VQE} and \texttt{BoolSat}, we perform
two additional experiments that measure the number of rounds and speedup with
respect to input size.

\begin{figure}[t]
	\centering
		\centering
		\includegraphics[width=\columnwidth]{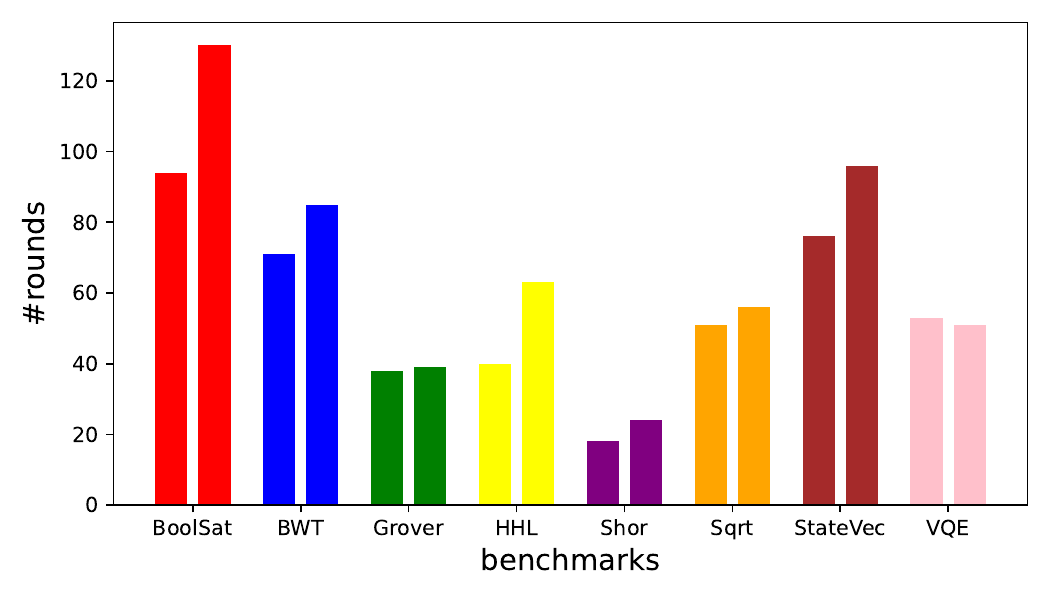}
		\caption{Number of rounds for different benchmarks. For each benchmark, the first bar represents the
			number of rounds for the smallest instance and the second bar represents
			the number of rounds for the largest instance.}
		\label{fig:n_rounds}
\end{figure}

\begin{figure}[t]
		\centering
		\includegraphics[width=\columnwidth]{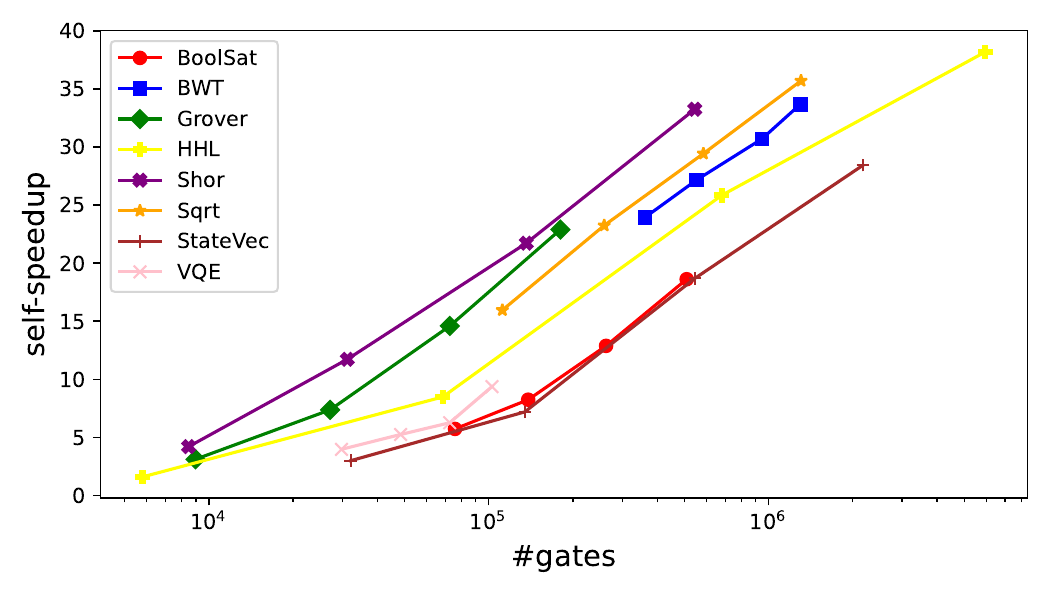}
		\caption{Self-speedup with respect to the number of gates using 64 threads. Each data point corresponds to a benchmark circuit. Log-scale on the x-axis.}
		\label{fig:speedup_vs_n_gates}
\end{figure}

%
Figure~\ref{fig:n_rounds} shows the number of rounds for each benchmark; for
each benchmark, the first bar is the number of rounds for the smallest instance
and the second bar is the number of rounds for the largest instance.
We observe that the number of rounds is between $20$ and $130$.
We also observe that for a given benchmark, the number of rounds increases
slightly compared to the input sizes, e.g.,
for \texttt{BoolSat}, the input size difference between the two bars is
approximately 8-fold, but the number of rounds increases by about
40\%.
This is important, because our theoretical analysis shows that our algorithm
has span linear in the number of rounds, which can be as large as the
number of gates, but in practice, we see that the number of rounds is smaller.
This is because practical quantum programs/circuits exhibit a great degree of
locality much like classical programs, and consist of independent sections that
do not interact deeply.
Coming back to the analysis of the scalability, we attribute the relatively low
scalability of \texttt{BoolSat} to the fact that it requires a large number of
rounds.

%

%


%
%

%
Figure~\ref{fig:speedup_vs_n_gates} shows the self-speedup
with 64 cores for varying input circuit sizes.
The figure shows that speedups increase with circuit size
and thus larger circuits provide more opportunities for parallel
optimization, where parallelism is most needed and beneficial.
%
%
We explain the poor scalability of \texttt{VQE} partly by the fact that its
circuit sizes are small (as can be seen in
Figure~\ref{fig:speedup_vs_n_gates}).
This does not explain, however, the slight degradation in the scalability of
\texttt{VQE} after 32 cores.
Poor scalability due to difficulties in controlling the grain of
parallelism is a common problem when problem sizes are small.
It is possible to engineer around this problem by carefully optimizing
the code on the target machine\cite{acg-heartbeat-2018}.
In our experiments, we have instead allowed the Rust language to
manage parallelism automatically; the experiment with the \texttt{VQE}
benchmark shows that Rust does mostly a good job but leaves some room
for improvement.


\subsection{Our Optimizer is Work Efficient}
\label{sec:exp::work}

The term work efficiency of a parallel algorithm refers to its efficiency with
respect to the work of an optimized sequential algorithm solving the same
problem.
Theoretically, our algorithm performs $O(n\lg{n})$ work (assuming constant
$\Omega$) and is therefore reasonably work efficient, with respect to the
$\Omega(n)$ lower bound that would be needed for circuit optimization.
To assess practical work efficiency, we compare our optimizer to the
\oac{}\cite{oac2025} optimizer, which is the fastest sequential optimizer
available.
\oac{} also ensures local optimality (as our algorithm does), making it
an appropriate baseline for this comparison.
In this comparison, both optimizers use the same oracle (\voqc{}) and use
approximately the same $\Omega$ value to ensure the optimization quality
differences are within $0.1\%$ of each other.
%

As shown in Table~\ref{tab:compare_with_oac},
in a vast majority of the benchmarks our optimizer outperforms \oac{}, usually
significantly.
For example, for the \texttt{HHL} benchmark with 13 qubits, \name{} is
approximately $5\times$ faster than \oac{}.
In the case of a few benchmarks, \texttt{HHL}, \texttt{StateVec}, and
\texttt{VQE}, our optimizer performs slightly slower with the smallest circuits
but outperforms \oac{} with all other circuits.
We attribute our performance advantage over \oac{} to the \oac{}'s quadratic
overhead from cutting and melding circuits during optimization, which our
index-tree data structure efficiently avoids.
As a result of this asymptotic gap, the performance advantage of our optimizer
widens over \oac{} as circuit size increases.
We note that it is interesting that our parallel algorithm when run on a
uniprocessor outperforms the best sequential optimizer.
This shows that with the right algorithm, the (perceived) overheads of
parallelism can be compensated to reap its benefits.
%
\begin{table*}[ht]
	\begin{tabular*}{\textwidth}{@{\extracolsep{\fill}} cccccc}
		                                   &           & \multicolumn{2}{c}{Time(s)} & \multicolumn{2}{c}{Gate Reduction}                              \\ \cmidrule(lr){3-4} \cmidrule(lr){5-6}
		                                   & n\_qubits & \oac{}                      & \name{}(1-thread)                  & \oac{} & \name{}(1-thread) \\ \midrule
		\multirow{4}{*}{\texttt{BoolSat}}  & 28        & 51.68                       & 44.32                              & 83.8\% & 83.7\%            \\
		                                   & 30        & 94.00                       & 75.97                              & 83.7\% & 83.6\%            \\
		                                   & 32        & 260.95                      & 151.01                             & 83.5\% & 83.4\%            \\
		                                   & 34        & 571.49                      & 293.09                             & 83.4\% & 83.4\%            \\
		\midrule
		\multirow{4}{*}{\texttt{BWT}}      & 17        & 734.48                      & 298.10                             & 45.1\% & 45.0\%            \\
		                                   & 21        & 1392.15                     & 577.99                             & 52.2\% & 52.2\%            \\
		                                   & 25        & 2398.93                     & 1083.17                            & 55.4\% & 55.2\%            \\
		                                   & 29        & 4632.67                     & 1684.60                            & 56.5\% & 56.2\%            \\
		\midrule
		\multirow{4}{*}{\texttt{Grover}}   & 9         & 5.73                        & 5.51                               & 29.4\% & 29.4\%            \\
		                                   & 11        & 23.54                       & 20.13                              & 30.0\% & 30.0\%            \\
		                                   & 13        & 89.72                       & 52.49                              & 29.8\% & 29.8\%            \\
		                                   & 15        & 311.93                      & 152.74                             & 29.6\% & 29.5\%            \\
		\midrule
		\multirow{4}{*}{\texttt{HHL}}      & 7         & 1.39                        & 1.55                               & 59.0\% & 58.9\%            \\
		                                   & 9         & 29.47                       & 27.82                              & 59.5\% & 59.5\%            \\
		                                   & 11        & 785.73                      & 316.79                             & 56.6\% & 56.6\%            \\
		                                   & 13        & 17968.66                    & 2692.69                            & 56.1\% & 56.0\%            \\
		\midrule
		\multirow{4}{*}{\texttt{Shor}}     & 10        & 6.01                        & 5.33                               & 11.0\% & 11.0\%            \\
		                                   & 12        & 25.41                       & 14.48                              & 3.2\%  & 3.2\%             \\
		                                   & 14        & 154.12                      & 70.52                              & 11.2\% & 11.2\%            \\
		                                   & 16        & 968.38                      & 274.75                             & 11.2\% & 11.2\%            \\
		\midrule
		\multirow{4}{*}{\texttt{Sqrt}}     & 42        & 156.25                      & 83.26                              & 42.1\% & 41.9\%            \\
		                                   & 48        & 440.49                      & 215.53                             & 42.2\% & 41.8\%            \\
		                                   & 54        & 1306.68                     & 507.52                             & 42.2\% & 41.8\%            \\
		                                   & 60        & 3592.70                     & 1243.65                            & 42.2\% & 41.8\%            \\
		\midrule
		\multirow{4}{*}{\texttt{StateVec}} & 5         & 4.28                        & 6.12                               & 79.6\% & 79.6\%            \\
		                                   & 6         & 33.23                       & 31.35                              & 79.2\% & 79.2\%            \\
		                                   & 7         & 207.16                      & 133.82                             & 78.8\% & 78.8\%            \\
		                                   & 8         & 1393.07                     & 488.68                             & 78.7\% & 78.7\%            \\
		\midrule
		\multirow{4}{*}{\texttt{VQE}}      & 18        & 4.60                        & 4.89                               & 64.8\% & 64.8\%            \\
		                                   & 22        & 8.99                        & 8.66                               & 61.9\% & 62.0\%            \\
		                                   & 26        & 20.67                       & 12.72                              & 59.3\% & 59.3\%            \\
		                                   & 30        & 31.97                       & 18.78                              & 56.9\% & 56.9\%            \\
		\bottomrule
	\end{tabular*}
	\caption{Optimization quality (represented as gate reduction) and running time comparison of \name{} and \texttt{OAC}. For fairness, we execute \name{} on a single thread and increase $\Omega$ to 400.}
	\label{tab:compare_with_oac}
\end{table*}


As another measure of work efficiency, we have also measured the fraction of
the time our optimizer spends within the oracle, doing actual optimizations, as
opposed to ``administrative'' work, including selecting and updating fingers.
As discussed in the \iffull Appendix (Section~\ref{sec:exp::oracletime}), \else
	extended version\cite{liu2025popqc}, \fi our optimizer spends over 90\% of its time in the oracle, showing
that very little time is spent for administrative purposes.

\subsection{Our Optimizer is Flexible}
\label{sec:exp::quartz}

\begin{figure}[t]
	\centering
	\includegraphics[width=\columnwidth]{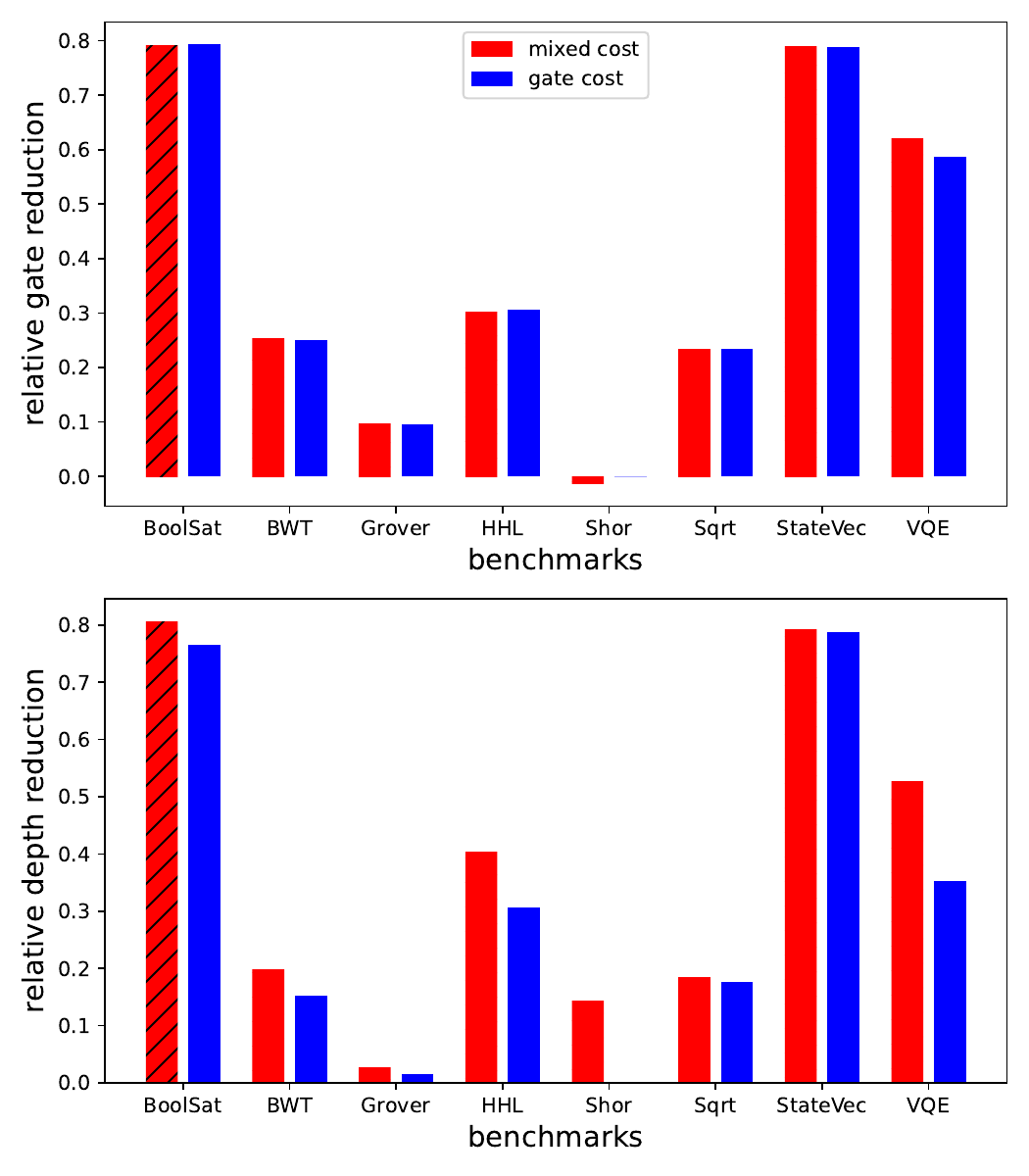}
	\caption{Gate and depth reduction using \quartz{} with different cost functions. Each bar corresponds to the average of the four instances in each circuit family.}
	\label{fig:quartz}
\end{figure}

To demonstrate the flexibility of our method, we evaluate it using
\quartz{}\cite{quartz-2022} as another oracle optimizer and a layered circuit
representation. \quartz{} is a search-based optimizer that supports
customizable cost functions to guide optimization. We define a cost function
that balances circuit depth and gate count, with a stronger emphasis on depth
reduction: $\text{cost} = 10\times \text{depth} + \text{gates}$. This cost
function is used both by \quartz{} and by our algorithm when deciding whether
to accept the oracle's optimizations. For efficient cost computation, we
represent circuits in layers and optimize at the layer granularity with
$\Omega=100$. The results are shown in Figure~\ref{fig:quartz}.

Our experiments show that using this depth-aware cost function achieves
significantly better depth reduction compared to optimizing purely for gate
count, with only modest increases in gate count. Two benchmarks demonstrate
particularly interesting results: For \texttt{Shor}, while \quartz{} with a
pure gate count objective finds no optimizations, our depth-aware approach
reduces circuit depth by 20\% with only a small gate count increase. For
\texttt{VQE}, the depth-aware cost function improves both depth and gate count
compared to gate-only optimization. We hypothesize that by encouraging more
compact gate arrangements, the depth-aware cost function helps create
additional optimization opportunities that can be exploited in subsequent
rounds.

\section{Related Work}
\label{sec:related}

In this section, we provide a comprehensive overview of quantum circuit
optimization techniques. Our approach is orthogonal to the techniques discussed
in this section, and can be combined with them to achieve both the benefits of
local optimality and the benefits of the optimization techniques.

\subsection{Optimization Objectives}
In addition to gate count, other optimization metrics have been proposed.
In the NISQ era, circuit fidelity, which is directly related to the success probability of quantum computation, is one of the most important metrics.
Notable examples of these cost functions include noise-resilient circuit
optimization to maximize fidelity\cite{murali2019noise, tannu2019not}. Also,
NISQ devices have various hardware constraints, and thus the optimization
should be tailored to specific device architectures, like topology, gateset, and
pulse. Some examples include Ref.\cite{molavi2022qubit, lye2015determining,
	itoko2020optimization, li2019tackling, nottingham2023decomposing, wu2021tilt,
	shi2019optimized, gokhale2020optimized, davis2020towards}.

In the fault-tolerant quantum computing era, it is widely believed that the
$\texttt{T}$ gate is the most expensive gate, and thus the $\texttt{T}$ count
becomes the most important metric in the fault-tolerant era. For small
circuits, algorithms for generating asymptotically optimal $\texttt{T}$ count
circuits\cite{giles2013remarks} exist but are not efficient for large
circuits.
There are also optimizers targeted at minimizing the $\texttt{T}$ count with
more efficient algorithms\cite{amy2014polynomial,amy2019formal}.
%
%
\subsection{Optimization Techniques}
To achieve the aforementioned optimization objectives, many optimization
techniques have been proposed. We classify them into rule-based, search-based,
resynthesis-based, and approximated methods.
\paragraph{Rule-based methods.}
%
Rule-based methods encode heuristic rules into the optimizer and apply them to
the circuit to optimize the circuit\cite{iten2022exact,
	bandyopadhyay2020post,hietala2021verified, quartz-2022}.
The optimizer proposed by Nam et al.\cite{Nam_2018} and the following
formally-verified implementation\cite{hietala2021verified} are great examples
of rule-based optimization.
%
%
These rules take quadratic time or even cubic time in circuit
size\cite{Nam_2018}, thus they are not suitable for large circuits.
%
%
%
%
As another example, \pyzx{}\cite{kissinger2020Pyzx} converts the circuit into ZX-diagrams and
applies ZX-calculus rules to optimize the circuit which is another example of
rule-based optimization.
\paragraph{Resynthesis-based methods.}
Resynthesis methods compute the unitary matrix of a small sequence of gates,
then use mathematical techniques to decompose the unitary matrix into some
gateset and hope the decomposed circuit is better than the original one.
Examples include KAK decomposition\cite{tucci2005introduction} and Cartan
decomposition\cite{khaneja2001cartan}.
%
%
However, these techniques can't scale to larger circuits: before decomposing,
even computing the unitary takes exponential time in circuit size.
QGo\cite{wu2020qgo} proposes a hierarchical approach that partitions the
circuit into blocks and resynthesizes each block to address the scalability
issue.
%

\paragraph{Search-based methods.}
Rule-based optimizers often contain manually designed rules that have limited optimization capability and are not flexible enough for gateset and cost function.
To address this, search-based optimizers have been proposed\cite{queso-2023,
	quartz-2022, qfast,qsearch} that automatically synthesize rules to optimize the
circuit.
%
%
%
%
%
%
A greedy application of the rules might lead to a local minimum. To avoid this, \quartz{}\cite{quartz-2022} and \queso{}\cite{queso-2023} search for all possible rules that can be applied to the circuit, even if they increase the gate count, hoping to find better optimizations in future iterations. 
%
As a result, this approach has delivered excellent reductions in gate count for
relatively small benchmarks but struggled to scale to large circuits.
%
%
Further improvements in this direction include reinforcement learning-based
methods\cite{li2024quarl,fosel2021quantum} that guide the optimization of
quantum circuits, and combining rule-based, search-based, and resynthesis
methods\cite{xu2024optimizing} to deliver a better trade-off between the
circuit quality and the optimization time.

\paragraph{Approximated optimization methods.}
The above optimization methods focus on finding a better circuit with the exact
same unitary. However, in practice, quantum computers can tolerate a small
error in the unitary\cite{nielsen2002quantum}, and algorithms like quantum
machine learning and quantum variational algorithms can tolerate even more. As
a result, researchers have also developed approximated optimization methods.
\textsf{QFast}\cite{qfast} and \textsf{QSearch}\cite{qsearch} apply numerical optimizations to
search for circuit decompositions that are close to the desired unitary.
Researchers have also developed machine learning-based methods for optimizing
specific quantum
circuits\cite{wang2022quantumnas,sim2021adaptive,ostaszewski2021reinforcement}
for variational or quantum machine learning applications.

%

\section{Discussion and Conclusion}
\label{sec:conclusion}

This paper presents a parallel algorithm, called \name, for circuit
optimization and an implementation of the algorithm.
The algorithm is reasonably work efficient and incurs only a logarithmic factor
overhead over the lower bound and works well in practice.
The algorithm guarantees a notion of local optimality with respect to the
oracle used for optimizing small segments of the circuit.
Due to its efficiency, performance, and quality guarantees, our implementation
delivers significant speedups over existing optimizers without degrading
optimization quality.

\section*{Acknowledgments}
This research was supported by the following NSF grants CCF-1901381,
CCF-2115104, CCF-2119352, CCF-2107241.  We are grateful to Chameleon
Cloud for providing the compute cycles needed for the experiments.

\bibliographystyle{plain}

\bibliography{Ref_QCS_Ding,local}

\begin{thebibliography}{10}

\bibitem{ab-book-algorithms}
Umut~A. Acar and Guy~E. Blelloch.
\newblock {\em Algorithms: Parallel and Sequential}.
\newblock 2022.
\newblock \url{http:www.algorithms-book.com}.

\bibitem{acg-heartbeat-2018}
Umut~A. Acar, Arthur Chargu{\'e}raud, Adrien Guatto, Mike Rainey, and Filip
  Sieczkowski.
\newblock Heartbeat scheduling: Provable efficiency for nested parallelism.
\newblock In {\em Proceedings of the 39th ACM SIGPLAN Conference on Programming
  Language Design and Implementation}, PLDI 2018, pages 769--782, 2018.

\bibitem{acarchra11}
Umut~A. Acar, Arthur Chargu{\'e}raud, and Mike Rainey.
\newblock Oracle scheduling: Controlling granularity in implicitly parallel
  languages.
\newblock In {\em {ACM SIGPLAN} Conference on Object-Oriented Programming,
  Systems, Languages, and Applications ({OOPSLA})}, pages 499--518, 2011.

\bibitem{amy2019formal}
Matthew Amy.
\newblock Formal methods in quantum circuit design.
\newblock 2019.

\bibitem{amy2014polynomial}
Matthew Amy, Dmitri Maslov, and Michele Mosca.
\newblock Polynomial-time t-depth optimization of clifford+ t circuits via
  matroid partitioning.
\newblock {\em IEEE Transactions on Computer-Aided Design of Integrated
  Circuits and Systems}, 33(10):1476--1489, 2014.

\bibitem{awa+space-2021}
Jatin Arora, Sam Westrick, and Umut~A. Acar.
\newblock Provably space efficient parallel functional programming.
\newblock In {\em Proceedings of the 48th Annual {ACM} Symposium on Principles
  of Programming Languages (POPL)}, 2021.

\bibitem{awa+ent-2023}
Jatin Arora, Sam Westrick, and Umut~A. Acar.
\newblock Efficient parallel functional programming with effects.
\newblock {\em Proc. {ACM} Program. Lang.}, 7({PLDI}):1558--1583, 2023.

\bibitem{oac2025}
Jatin Arora, Mingkuan Xu, Sam Westrick, Pengyu Liu, Dantong Li, Yongshan Ding,
  and Umut~A Acar.
\newblock Local optimization of quantum circuits (extended version).
\newblock {\em arXiv preprint arXiv:2502.19526}, 2025.

\bibitem{bandyopadhyay2020post}
Chandan Bandyopadhyay, Robert Wille, Rolf Drechsler, and Hafizur Rahaman.
\newblock Post synthesis-optimization of reversible circuit using template
  matching.
\newblock In {\em 2020 24th International Symposium on VLSI Design and Test
  (VDAT)}, pages 1--4. IEEE, 2020.

\bibitem{Benioff80}
Paul Benioff.
\newblock The computer as a physical system: A microscopic quantum mechanical
  hamiltonian model of computers as represented by turing machines.
\newblock {\em Journal of Statistical Physics}, 22:563--591, 05 1980.

\bibitem{bergholm2018pennylane}
Ville Bergholm, Josh Izaac, Maria Schuld, Christian Gogolin, Shahnawaz Ahmed,
  Vishnu Ajith, M~Sohaib Alam, Guillermo Alonso-Linaje, B~AkashNarayanan, Ali
  Asadi, et~al.
\newblock Pennylane: Automatic differentiation of hybrid quantum-classical
  computations.
\newblock {\em arXiv preprint arXiv:1811.04968}, 2018.

\bibitem{biamonte2017quantum}
Jacob Biamonte, Peter Wittek, Nicola Pancotti, Patrick Rebentrost, Nathan
  Wiebe, and Seth Lloyd.
\newblock Quantum machine learning.
\newblock {\em Nature}, 549(7671):195--202, 2017.

\bibitem{childs2017quantum}
Andrew~M Childs, Robin Kothari, and Rolando~D Somma.
\newblock Quantum algorithm for systems of linear equations with exponentially
  improved dependence on precision.
\newblock {\em SIAM Journal on Computing}, 46(6):1920--1950, 2017.

\bibitem{davis2020towards}
Marc~G Davis, Ethan Smith, Ana Tudor, Koushik Sen, Irfan Siddiqi, and Costin
  Iancu.
\newblock Towards optimal topology aware quantum circuit synthesis.
\newblock In {\em 2020 IEEE International Conference on Quantum Computing and
  Engineering (QCE)}, pages 223--234. IEEE, 2020.

\bibitem{ebadi2021quantum}
Sepehr Ebadi, Tout~T Wang, Harry Levine, Alexander Keesling, Giulia Semeghini,
  Ahmed Omran, Dolev Bluvstein, Rhine Samajdar, Hannes Pichler, Wen~Wei Ho,
  et~al.
\newblock Quantum phases of matter on a 256-atom programmable quantum
  simulator.
\newblock {\em Nature}, 595(7866):227--232, 2021.

\bibitem{Feynman82}
Richard~P Feynman.
\newblock Simulating physics with computers.
\newblock In {\em Feynman and computation}, pages 133--153. CRC Press, 2018.

\bibitem{fosel2021quantum}
Thomas F{\"o}sel, Murphy~Yuezhen Niu, Florian Marquardt, and Li~Li.
\newblock Quantum circuit optimization with deep reinforcement learning.
\newblock {\em arXiv preprint arXiv:2103.07585}, 2021.

\bibitem{giles2013remarks}
Brett Giles and Peter Selinger.
\newblock Remarks on matsumoto and amano's normal form for single-qubit
  clifford+ t operators.
\newblock {\em arXiv preprint arXiv:1312.6584}, 2013.

\bibitem{gokhale2020optimized}
Pranav Gokhale, Ali Javadi-Abhari, Nathan Earnest, Yunong Shi, and Frederic~T
  Chong.
\newblock Optimized quantum compilation for near-term algorithms with
  openpulse.
\newblock In {\em 2020 53rd Annual IEEE/ACM International Symposium on
  Microarchitecture (MICRO)}, pages 186--200. IEEE, 2020.

\bibitem{grover1996fast}
Lov~K Grover.
\newblock A fast quantum mechanical algorithm for database search.
\newblock {\em arXiv preprint quant-ph/9605043}, 1996.

\bibitem{harrow2009quantum}
Aram~W Harrow, Avinatan Hassidim, and Seth Lloyd.
\newblock Quantum algorithm for linear systems of equations.
\newblock {\em Physical review letters}, 103(15):150502, 2009.

\bibitem{hietala2021verified}
Kesha Hietala, Robert Rand, Shih-Han Hung, Xiaodi Wu, and Michael Hicks.
\newblock A verified optimizer for quantum circuits.
\newblock {\em Proceedings of the ACM on Programming Languages}, 5(POPL):1--29,
  2021.

\bibitem{iten2022exact}
Raban Iten, Romain Moyard, Tony Metger, David Sutter, and Stefan Woerner.
\newblock Exact and practical pattern matching for quantum circuit
  optimization.
\newblock {\em ACM Transactions on Quantum Computing}, 3(1):1--41, 2022.

\bibitem{itoko2020optimization}
Toshinari Itoko, Rudy Raymond, Takashi Imamichi, and Atsushi Matsuo.
\newblock Optimization of quantum circuit mapping using gate transformation and
  commutation.
\newblock {\em Integration}, 70:43--50, 2020.

\bibitem{janzing2003identity}
Dominik Janzing, Pawel Wocjan, and Thomas Beth.
\newblock Identity check is qma-complete, 2003.

\bibitem{khaneja2001cartan}
Navin Khaneja and Steffen~J Glaser.
\newblock Cartan decomposition of su (2n) and control of spin systems.
\newblock {\em Chemical Physics}, 267(1-3):11--23, 2001.

\bibitem{kissinger2020Pyzx}
Aleks Kissinger and John van~de Wetering.
\newblock {PyZX: Large Scale Automated Diagrammatic Reasoning}.
\newblock In Bob Coecke and Matthew Leifer, editors, {\em {\rm Proceedings 16th
  International Conference on} Quantum Physics and Logic, {\rm Chapman
  University, Orange, CA, USA., 10-14 June 2019}}, volume 318 of {\em
  Electronic Proceedings in Theoretical Computer Science}, pages 229--241. Open
  Publishing Association, 2020.

\bibitem{kjaergaard2020superconducting}
Morten Kjaergaard, Mollie~E Schwartz, Jochen Braum{\"u}ller, Philip Krantz,
  Joel I-J Wang, Simon Gustavsson, and William~D Oliver.
\newblock Superconducting qubits: Current state of play.
\newblock {\em Annual Review of Condensed Matter Physics}, 11(1):369--395,
  2020.

\bibitem{qsearch}
Costin Lancu, Marc Davis, Ethan Smith, and USDOE.
\newblock Quantum search compiler (qsearch) v2.0, version v2.0, 10 2020.

\bibitem{li2021qasmbench}
Ang Li, Samuel Stein, Sriram Krishnamoorthy, and James Ang.
\newblock Qasmbench: A low-level qasm benchmark suite for nisq evaluation and
  simulation.
\newblock {\em arXiv preprint arXiv:2005.13018}, 2021.

\bibitem{li2019tackling}
Gushu Li, Yufei Ding, and Yuan Xie.
\newblock Tackling the qubit mapping problem for nisq-era quantum devices.
\newblock In {\em Proceedings of the Twenty-Fourth International Conference on
  Architectural Support for Programming Languages and Operating Systems}, pages
  1001--1014, 2019.

\bibitem{li2024quarl}
Zikun Li, Jinjun Peng, Yixuan Mei, Sina Lin, Yi~Wu, Oded Padon, and Zhihao Jia.
\newblock Quarl: A learning-based quantum circuit optimizer.
\newblock {\em Proceedings of the ACM on Programming Languages},
  8(OOPSLA1):555--582, 2024.

\bibitem{lye2015determining}
Aaron Lye, Robert Wille, and Rolf Drechsler.
\newblock Determining the minimal number of swap gates for multi-dimensional
  nearest neighbor quantum circuits.
\newblock In {\em The 20th Asia and South Pacific Design Automation
  Conference}, pages 178--183. IEEE, 2015.

\bibitem{molavi2022qubit}
Abtin Molavi, Amanda Xu, Martin Diges, Lauren Pick, Swamit Tannu, and Aws
  Albarghouthi.
\newblock Qubit mapping and routing via maxsat.
\newblock In {\em 2022 55th IEEE/ACM International Symposium on
  Microarchitecture (MICRO)}, pages 1078--1091. IEEE, 2022.

\bibitem{monroe2021programmable}
Christopher Monroe, Wes~C Campbell, L-M Duan, Z-X Gong, Alexey~V Gorshkov,
  Paul~W Hess, Rajibul Islam, Kihwan Kim, Norbert~M Linke, Guido Pagano, et~al.
\newblock Programmable quantum simulations of spin systems with trapped ions.
\newblock {\em Reviews of Modern Physics}, 93(2):025001, 2021.

\bibitem{moses2023race}
S.~A. Moses, C.~H. Baldwin, M.~S. Allman, R.~Ancona, L.~Ascarrunz, C.~Barnes,
  J.~Bartolotta, B.~Bjork, P.~Blanchard, M.~Bohn, J.~G. Bohnet, N.~C. Brown,
  N.~Q. Burdick, W.~C. Burton, S.~L. Campbell, J.~P. Campora~III au2,
  C.~Carron, J.~Chambers, J.~W. Chan, Y.~H. Chen, A.~Chernoguzov, E.~Chertkov,
  J.~Colina, J.~P. Curtis, R.~Daniel, M.~DeCross, D.~Deen, C.~Delaney, J.~M.
  Dreiling, C.~T. Ertsgaard, J.~Esposito, B.~Estey, M.~Fabrikant, C.~Figgatt,
  C.~Foltz, M.~Foss-Feig, D.~Francois, J.~P. Gaebler, T.~M. Gatterman, C.~N.
  Gilbreth, J.~Giles, E.~Glynn, A.~Hall, A.~M. Hankin, A.~Hansen, D.~Hayes,
  B.~Higashi, I.~M. Hoffman, B.~Horning, J.~J. Hout, R.~Jacobs, J.~Johansen,
  L.~Jones, J.~Karcz, T.~Klein, P.~Lauria, P.~Lee, D.~Liefer, C.~Lytle, S.~T.
  Lu, D.~Lucchetti, A.~Malm, M.~Matheny, B.~Mathewson, K.~Mayer, D.~B. Miller,
  M.~Mills, B.~Neyenhuis, L.~Nugent, S.~Olson, J.~Parks, G.~N. Price, Z.~Price,
  M.~Pugh, A.~Ransford, A.~P. Reed, C.~Roman, M.~Rowe, C.~Ryan-Anderson,
  S.~Sanders, J.~Sedlacek, P.~Shevchuk, P.~Siegfried, T.~Skripka, B.~Spaun,
  R.~T. Sprenkle, R.~P. Stutz, M.~Swallows, R.~I. Tobey, A.~Tran, T.~Tran,
  E.~Vogt, C.~Volin, J.~Walker, A.~M. Zolot, and J.~M. Pino.
\newblock A race track trapped-ion quantum processor, 2023.

\bibitem{murali2019noise}
Prakash Murali, Jonathan~M Baker, Ali Javadi-Abhari, Frederic~T Chong, and
  Margaret Martonosi.
\newblock Noise-adaptive compiler mappings for noisy intermediate-scale quantum
  computers.
\newblock In {\em Proceedings of the Twenty-Fourth International Conference on
  Architectural Support for Programming Languages and Operating Systems}, pages
  1015--1029. ACM, 2019.

\bibitem{Nam_2018}
Yunseong Nam, Neil~J. Ross, Yuan Su, Andrew~M. Childs, and Dmitri Maslov.
\newblock Automated optimization of large quantum circuits with continuous
  parameters.
\newblock {\em npj Quantum Information}, 4(1), may 2018.

\bibitem{nielsen2002quantum}
Michael~A Nielsen and Isaac Chuang.
\newblock Quantum computation and quantum information, 2002.

\bibitem{nottingham2023decomposing}
Natalia Nottingham, Michael~A Perlin, Ryan White, Hannes Bernien, Frederic~T
  Chong, and Jonathan~M Baker.
\newblock Decomposing and routing quantum circuits under constraints for
  neutral atom architectures.
\newblock {\em arXiv preprint arXiv:2307.14996}, 2023.

\bibitem{ostaszewski2021reinforcement}
Mateusz Ostaszewski, Lea~M Trenkwalder, Wojciech Masarczyk, Eleanor Scerri, and
  Vedran Dunjko.
\newblock Reinforcement learning for optimization of variational quantum
  circuit architectures.
\newblock {\em Advances in Neural Information Processing Systems},
  34:18182--18194, 2021.

\bibitem{peruzzo2014variational}
Alberto Peruzzo, Jarrod McClean, Peter Shadbolt, Man-Hong Yung, Xiao-Qi Zhou,
  Peter~J Love, Al{\'a}n Aspuru-Guzik, and Jeremy~L O’brien.
\newblock A variational eigenvalue solver on a photonic quantum processor.
\newblock {\em Nature communications}, 5:4213, 2014.

\bibitem{preskill2018quantum}
John Preskill.
\newblock Quantum computing in the nisq era and beyond.
\newblock {\em Quantum}, 2:79, 2018.

\bibitem{preskill2024beyond}
John Preskill.
\newblock Beyond nisq: The megaquop machine, 2024.

\bibitem{scholl2021quantum}
Pascal Scholl, Michael Schuler, Hannah~J Williams, Alexander~A Eberharter,
  Daniel Barredo, Kai-Niklas Schymik, Vincent Lienhard, Louis-Paul Henry,
  Thomas~C Lang, Thierry Lahaye, et~al.
\newblock Quantum simulation of 2d antiferromagnets with hundreds of rydberg
  atoms.
\newblock {\em Nature}, 595(7866):233--238, 2021.

\bibitem{schuld2015introduction}
Maria Schuld, Ilya Sinayskiy, and Francesco Petruccione.
\newblock An introduction to quantum machine learning.
\newblock {\em Contemporary Physics}, 56(2):172--185, 2015.

\bibitem{shi2019optimized}
Yunong Shi, Nelson Leung, Pranav Gokhale, Zane Rossi, David~I Schuster, Henry
  Hoffmann, and Frederic~T Chong.
\newblock Optimized compilation of aggregated instructions for realistic
  quantum computers.
\newblock In {\em Proceedings of the Twenty-Fourth International Conference on
  Architectural Support for Programming Languages and Operating Systems}, pages
  1031--1044. ACM, 2019.

\bibitem{shor1994algorithms}
Peter~W Shor.
\newblock Algorithms for quantum computation: Discrete logarithms and
  factoring.
\newblock In {\em Proceedings 35th annual symposium on foundations of computer
  science}, pages 124--134. Ieee, 1994.

\bibitem{sim2021adaptive}
Sukin Sim, Jonathan Romero, J{\'e}r{\^o}me~F Gonthier, and Alexander~A Kunitsa.
\newblock Adaptive pruning-based optimization of parameterized quantum
  circuits.
\newblock {\em Quantum Science and Technology}, 6(2):025019, 2021.

\bibitem{tannu2019not}
Swamit~S Tannu and Moinuddin~K Qureshi.
\newblock Not all qubits are created equal: a case for variability-aware
  policies for nisq-era quantum computers.
\newblock In {\em Proceedings of the Twenty-Fourth International Conference on
  Architectural Support for Programming Languages and Operating Systems}, pages
  987--999. ACM, 2019.

\bibitem{tucci2005introduction}
Robert~R Tucci.
\newblock An introduction to cartan's kak decomposition for qc programmers.
\newblock {\em arXiv preprint quant-ph/0507171}, 2005.

\bibitem{wang2022quantumnas}
Hanrui Wang, Yongshan Ding, Jiaqi Gu, Yujun Lin, David~Z Pan, Frederic~T Chong,
  and Song Han.
\newblock Quantumnas: Noise-adaptive search for robust quantum circuits.
\newblock In {\em 2022 IEEE International Symposium on High-Performance
  Computer Architecture (HPCA)}, pages 692--708. IEEE, 2022.

\bibitem{westrick+disent-2020}
Sam Westrick, Rohan Yadav, Matthew Fluet, and Umut~A. Acar.
\newblock Disentanglement in nested-parallel programs.
\newblock In {\em Proceedings of the 47th Annual {ACM} Symposium on Principles
  of Programming Languages (POPL)}, 2020.

\bibitem{qiskit-2019}
Robert Wille, Rod Van~Meter, and Yehuda Naveh.
\newblock Ibm’s qiskit tool chain: Working with and developing for real
  quantum computers.
\newblock In {\em 2019 Design, Automation \& Test in Europe Conference \&
  Exhibition (DATE)}, pages 1234--1240. IEEE, 2019.

\bibitem{wu2020qgo}
Xin-Chuan Wu, Marc~Grau Davis, Frederic~T Chong, and Costin Iancu.
\newblock Qgo: Scalable quantum circuit optimization using automated synthesis.
\newblock {\em arXiv preprint arXiv:2012.09835}, 2020.

\bibitem{wu2021tilt}
Xin-Chuan Wu, Dripto~M Debroy, Yongshan Ding, Jonathan~M Baker, Yuri Alexeev,
  Kenneth~R Brown, and Frederic~T Chong.
\newblock Tilt: Achieving higher fidelity on a trapped-ion linear-tape quantum
  computing architecture.
\newblock In {\em 2021 IEEE International Symposium on High-Performance
  Computer Architecture (HPCA)}, pages 153--166. IEEE, 2021.

\bibitem{queso-2023}
Amanda Xu, Abtin Molavi, Lauren Pick, Swamit Tannu, and Aws Albarghouthi.
\newblock Synthesizing quantum-circuit optimizers.
\newblock {\em Proc. ACM Program. Lang.}, 7(PLDI), jun 2023.

\bibitem{xu2024optimizing}
Amanda Xu, Abtin Molavi, Swamit Tannu, and Aws Albarghouthi.
\newblock Optimizing quantum circuits, fast and slow.
\newblock {\em arXiv preprint arXiv:2411.04104}, 2024.

\bibitem{quartz-2022}
Mingkuan Xu, Zikun Li, Oded Padon, Sina Lin, Jessica Pointing, Auguste Hirth,
  Henry Ma, Jens Palsberg, Alex Aiken, Umut~A. Acar, and Zhihao Jia.
\newblock Quartz: Superoptimization of quantum circuits.
\newblock In {\em Proceedings of the 43rd ACM SIGPLAN International Conference
  on Programming Language Design and Implementation}, PLDI 2022, page
  625–640, New York, NY, USA, 2022. Association for Computing Machinery.

\bibitem{qfast}
Ed~Younis, Koushik Sen, Katherine Yelick, and Costin Iancu.
\newblock Qfast: Conflating search and numerical optimization for scalable
  quantum circuit synthesis, 2021.

\end{thebibliography}

\iffull
\clearpage
\onecolumn
\appendix
\section{Additional Experiments}

\subsection{Our Optimizer Performs a Linear Number of Oracle Calls}

\begin{figure}[ht]
	\centering
	\includegraphics[width=0.7\columnwidth]{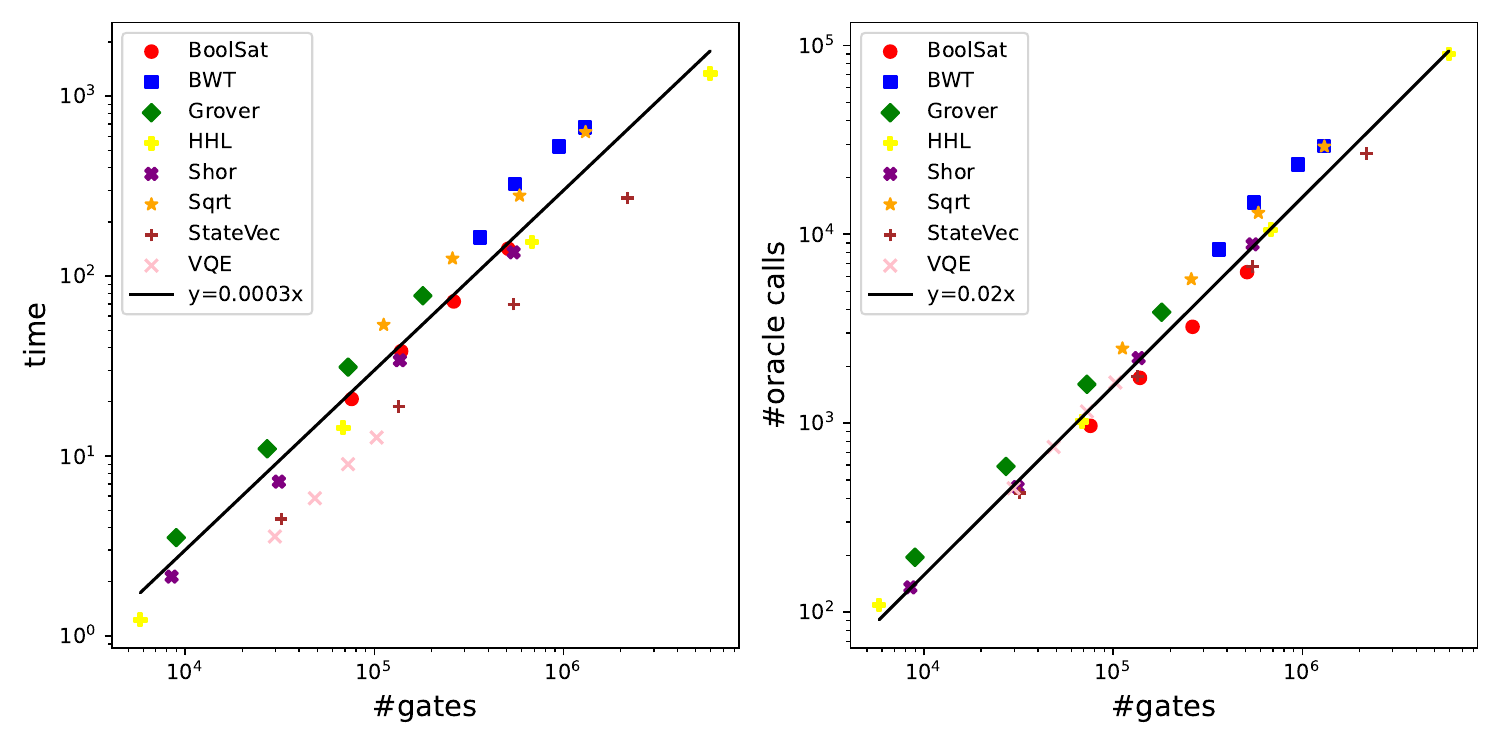}
	\caption{(Left) Work, measured by running time, vs. number of gates for different benchmarks. (Right) Number of oracle calls vs. number of gates for different benchmarks.}
	\label{fig:work_vs_gates}
\end{figure}

Figure~\ref{fig:work_vs_gates} shows the work (measured by the running time
using one thread) and the number of oracle calls vs. the number of gates for
different benchmarks. We see that the work and number of oracle calls are
approximately proportional to the number of gates, which is consistent with our
theoretical analysis.


\subsection{Our Optimizer Spends Most of its Time in the Oracle}
\label{sec:exp::oracletime}
\begin{figure}[t]
	\centering
	\includegraphics[width=0.7\columnwidth]{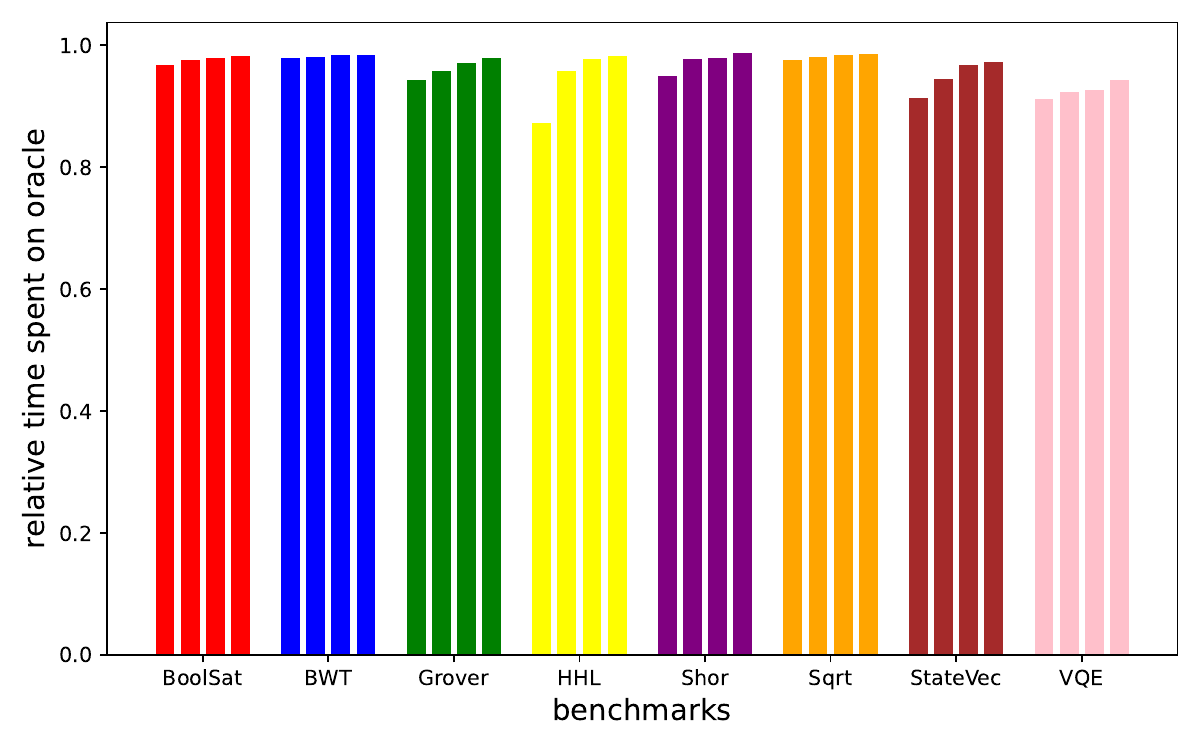}
	\caption{Percentage of oracle execution time across different benchmarks. Within each benchmark, bars are arranged from left to right in order of increasing instance size.}
	\label{fig:oracle_time}
\end{figure}

Figure~\ref{fig:oracle_time} shows the percentage of total execution time spent
on oracle calls across different benchmarks. For large circuits, oracle calls
consume over 90\% of the total runtime, with this percentage increasing as
circuit size grows. This demonstrates that our implementation overhead is
minimal, and the performance is primarily bounded by the underlying oracle
optimizer.

\subsection{Impact of $\Omega$}
\label{sec:exp::omega}
\begin{figure}[t]
	\centering
	\includegraphics[width=0.7\columnwidth]{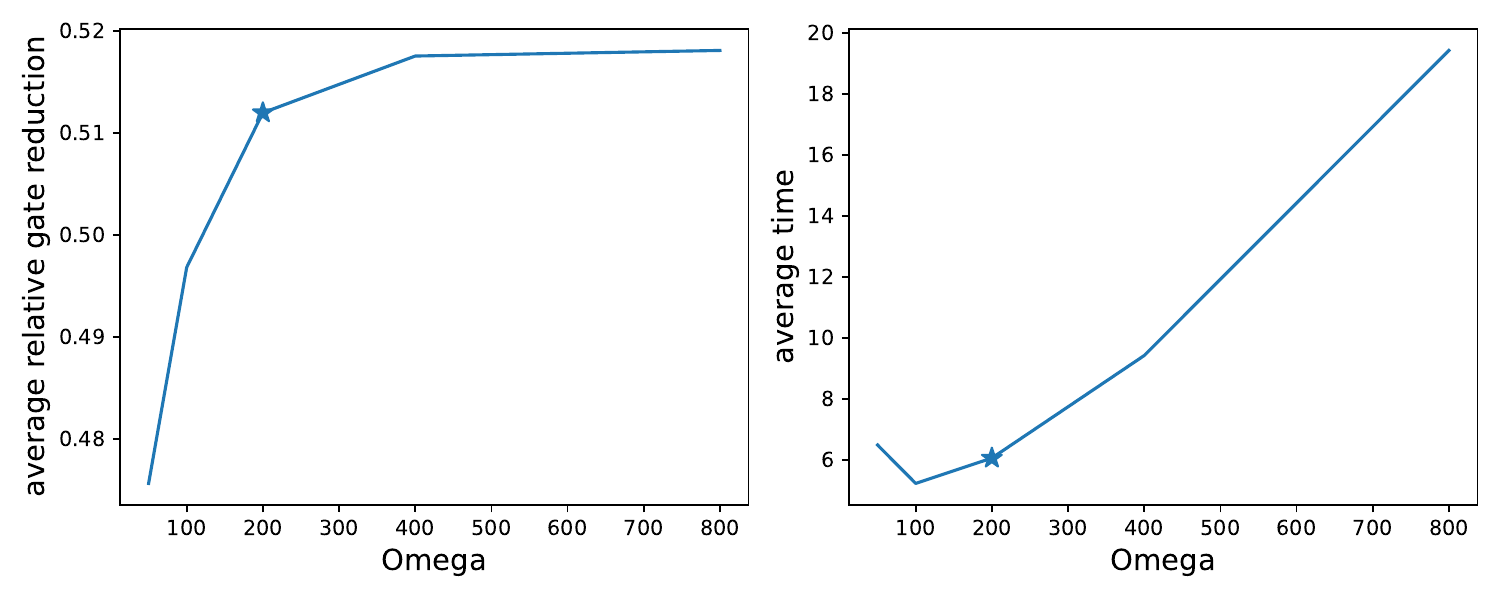}
	\caption{Impact of $\Omega$ on the performance of \name{}. Average across all benchmarks. The star indicates the default value of $\Omega=200$, which is used in all experiments unless explicitly stated otherwise.}
	\label{fig:omega}
\end{figure}

Figure~\ref{fig:omega} illustrates the impact of $\Omega$ on the quality and
efficiency of \name{}. The results show that optimization quality reaches its
maximum value when $\Omega$ exceeds $400$. Interestingly, the running time
decreases when $\Omega$ increases from $50$ to $100$. This occurs because at
very small values of $\Omega$, the overhead of oracle calls and administrative
work outweighs the actual optimization time.

The star in the figure indicates our default value of $\Omega=200$, which
provides a good balance between execution speed and optimization quality.
\subsection{Impact of Initial Circuit Ordering}

Initial circuit ordering can affect optimization performance. To quantify this effect, we performed additional experiments shown in Table~\ref{tab:ordering-impact}. We tested three different ordering strategies:

\begin{itemize}
\item \textbf{Left-justified}: Convert the circuit to a layered representation, push each gate as far left as possible, then convert back to the array-based representation.
\item \textbf{Right-justified}: Push gates as far right as possible using the same layered representation approach.
\item \textbf{Default Order}: Use the original ordering from the QASM file (the setting used throughout the paper).
\end{itemize}

\begin{table}[ht]
\centering
\begin{tabular*}{0.7\textwidth}{@{\extracolsep{\fill}}cccc}
\textbf{Benchmark} & \textbf{Left-justified} & \textbf{Right-justified} & \textbf{Default Order} \\\midrule
\texttt{BoolSat} & 83.52\% & 83.52\% & 83.52\% \\\midrule
\texttt{BWT} & 50.63\% & 50.52\% & 50.70\% \\\midrule
\texttt{Grover} & 29.28\% & 29.21\% & 29.27\% \\\midrule
\texttt{HHL} & 57.69\% & 57.68\% & 57.70\% \\\midrule
\texttt{Shor} & 9.10\% & 9.11\% & 9.10\% \\\midrule
\texttt{Sqrt} & 40.36\% & 41.58\% & 39.48\% \\\midrule
\texttt{StateVec} & 79.07\% & 79.05\% & 79.06\% \\\midrule
\texttt{VQE} & 60.70\% & 60.73\% & 60.75\% \\\midrule
\end{tabular*}
\caption{Average gate reduction across different benchmark families using \name{} with different initial orderings. Averages are calculated over all benchmark circuits of various sizes for each benchmark family (higher values indicate better optimization).}
\label{tab:ordering-impact}
\end{table}

The results demonstrate that, except for \texttt{Sqrt}, ordering differences result in performance variations of less than 0.2\%. We conclude that while ordering does influence results, the impact is relatively small for most circuits. 

The \texttt{Sqrt} benchmark demonstrates greater sensitivity to ordering because it contains many consecutive single-qubit gates. Positioning these gates near two-qubit gates acting on the same qubits reveals more optimization opportunities. More than 5\% of gates in \texttt{Sqrt} can be ``slided'' by more than 200 positions in the representation, while other benchmarks typically show less than 0.1\%. Consequently, our choice of $\Omega = 200$ causes some optimizable gates to be separated, leading to suboptimal quality. 

Based on this observation, a circuit-specific heuristic for choosing $\Omega$ is to set it accroding to the maximum sliding distance of gates in the circuit's array representation. And we leave this as a future work.

\fi
\end{document}